%% file: main.tex
\documentclass[11pt]{article}
\usepackage[margin=1in]{geometry}
\usepackage[T1]{fontenc}
\usepackage{lmodern}
\usepackage{amsmath,amssymb,amsthm,mathtools}
\usepackage{microtype}
\usepackage{graphicx,booktabs,array}
\usepackage{tikz,pgfplots}
\usetikzlibrary{arrows.meta,positioning,calc,fit,shapes.geometric}
\pgfplotsset{compat=1.18}
\usepackage{enumitem}
\usepackage{caption}
\usepackage{placeins}
\usepackage{xcolor}
\usepackage[colorlinks=true,linkcolor=blue!45!black,citecolor=blue!45!black,
urlcolor=blue!45!black,
  bookmarksnumbered=true,pdfusetitle]{hyperref}
\usepackage[nameinlink,noabbrev]{cleveref}
\numberwithin{equation}{section}
\newtheorem{theorem}{Theorem}[section]
\newtheorem{lemma}[theorem]{Lemma}
\newtheorem{proposition}[theorem]{Proposition}
\newtheorem{corollary}[theorem]{Corollary}
\theoremstyle{definition}

\newtheorem{example}[theorem]{Example}
\theoremstyle{remark}
\newtheorem{remark}[theorem]{Remark}
\DeclareMathOperator{\Tr}{Tr}
\DeclareMathOperator{\id}{id}

\newcommand{\E}{\mathbb E}
\newcommand{\N}{\mathbb N}
\newcommand{\cE}{\mathcal E}
\newcommand{\cA}{\mathcal A}
\newcommand{\cB}{\mathcal B}
\newcommand{\cM}{\mathcal M}
\newcommand{\cR}{\mathcal R}
\newcommand{\ket}[1]{|#1\rangle}

\newcommand{\proj}[1]{|#1\rangle\!\langle#1|}
\newcommand{\fb}{\mathrm{fb}}
\newcommand{\nf}{\mathrm{nf}}
\newcommand{\cl}{\mathrm{cl}}
\newcommand{\succp}{P_{\mathrm{succ}}}
\newcommand{\eps}{\varepsilon}
\newcommand{\binD}{D_{\mathrm{bin}}}
\newcommand{\Htest}{D_H}
\newcommand{\one}{\mathbf 1}
\tikzset{op/.style={draw,fill=white,rounded corners=2pt,minimum height=.66cm,
  minimum width=1.1cm,inner sep=4pt},
  qwire/.style={->,thick},cwire/.style={->,thick,dashed}}
\allowdisplaybreaks[2]
\title{Finite-blocklength classical communication over the quantum erasure channel with and without classical feedback}
\author{Mark M. Wilde\thanks{Email: \href{mailto:wilde@cornell.edu}{wilde@cornell.edu}}\\[3pt]
\textit{School of Electrical and Computer Engineering, Cornell University}\\
\textit{Ithaca, New York 14850, USA}}
\date{September 30, 2026}
\hypersetup{pdfauthor={Mark M. Wilde}}

\begin{document}
\maketitle
\begin{abstract}
We determine the optimal success probability for transmitting a fixed number of classical messages through a finite number of uses of the quantum erasure channel, assisted by noiseless classical feedback and without initial shared entanglement. For an input dimension $d$, an erasure probability $p$, a blocklength $n$, and $M$ equiprobable messages, the optimal success probability is equal to $\E[\min\{1,d^K/M\}]$, where $K$ is binomial with parameters $n$ and $1-p$. The converse allows arbitrary adaptive quantum encoders, quantum memories, and receiver instruments. Its main ingredient is an elementary dimension bound for noiseless quantum communication with classical feedback, proved by fixing the classical controls without conditioning the sender's state on the observed transcript. A classical protocol that transmits the base-\(d\) digits of an integer representing the message, repeating each digit until the receiver acknowledges its reception or the prescribed blocklength is reached, attains the bound for every integer \(M\). We also establish a relative-majorization property of erasure-channel outputs and exactly evaluate a hypothesis-testing converse, recovering the same numerical bound without feedback. That converse need not be achievable without feedback: four uses of the qubit erasure channel and four messages give a strict gap. A product-state code employing tetrahedral qubit states nevertheless outperforms every classical binary erasure code with these parameters. We give an explicit message-size formula and show that the bounded-remainder normal approximation and the average-success strong-converse exponent are unchanged without feedback. We also determine the feedback-assisted error exponent below capacity, prove that it agrees with the no-feedback exponent above a critical rate, and give no-feedback bounds at lower rates.
\end{abstract}

\tableofcontents

\section{Introduction}\label{sec:intro}

\subsection{Background and motivation}
The quantum erasure channel either transmits its input without alteration or replaces it by an orthogonal erasure flag. The receiver therefore knows which transmissions were lost. This simple structure makes the channel a useful setting in which to distinguish the effects of quantum encoding, collective measurements, and classical feedback. Its asymptotic capacities have long served as basic examples in quantum Shannon theory~\cite{BDS97,Wilde17}.

At a finite blocklength, the number of unerased transmissions fluctuates. A first dimension argument suggests that, if $k$ systems of dimension $d$ arrive, then the probability of correctly identifying one of $M$ equiprobable messages should be at most $\min\{1,d^k/M\}$. Without feedback, this follows by applying a state-discrimination dimension bound~\cite{NS06} to the reduced code states. With feedback, the argument is less immediate. The sender may retain quantum memory, the receiver may measure previously received systems and report the outcomes, and later channel inputs may depend on these outcomes. Moreover, the protocol may establish entanglement during its execution. One cannot simply replace an adaptive protocol by a fixed ensemble of received states without explaining how the feedback dependence is handled.

These observations lead to two questions. Does the same dimension bound remain valid under arbitrary classical feedback, including feedback obtained from measurements of quantum memories? And can a single communication protocol attain the bound for every erasure pattern? The second question is distinct from evaluating a converse expression: a converse can be exactly computable without being attained by an operational code. Understanding both questions is necessary for an exact finite-blocklength capacity formula.

\subsection{Summary of results}
Let $\cE^p$ be the $d$-dimensional quantum erasure channel with erasure probability $p$. For positive integers $n$ and $M$, let $\succp^{*,\fb}(n,M)$ denote the optimal average probability of correctly decoding one of $M$ equiprobable classical messages using exactly $n$ instances of this channel, assisted by noiseless classical feedback and without initial shared entanglement. The optimization allows arbitrary adaptive encoding and decoding operations, including quantum memories. The superscript $\fb$ indicates classical feedback assistance, while the asterisk indicates optimization over all such protocols. We suppress the dependence on the fixed channel parameters $d$ and $p$. Our main result, Theorem~\ref{thm:feedback}, is that
\begin{equation}
  \succp^{*,\fb}(n,M)
  =\E\!\left[\min\!\left\{1,\frac{d^K}{M}\right\}\right],
  \label{eq:intro-main}
\end{equation}
where $K$ is a binomial random variable with parameters $n$ and $1-p$, representing the number of unerased transmissions.

No free forward classical communication is permitted. The classical feedback may contain arbitrary receiver measurement outcomes, not only erasure acknowledgments. With an initial, message-independent random seed supplied through the feedback channel, the same value is attainable as the minimum success probability over all messages. Thus, the formula also determines the optimum under the maximum-error criterion.

The proof consists of a converse and an achievability argument. For the converse, we first study noiseless forward quantum communication with classical feedback, as illustrated in Figure~\ref{fig:noiseless}. Alice retains a quantum memory, and Bob's classical measurement outcomes control her later encoders. For each prescribed feedback transcript, we apply Alice's trace-preserving encoder channels using that transcript as their classical controls, without conditioning on its occurrence. This prepares a normalized state on the forward-transmitted systems. Bob's complete sequence of instruments defines a single positive operator-valued measure (POVM) whose outcomes include the transcript and his final estimate. The resulting trace pairing reproduces the original joint outcome probabilities, so a dimension bound applies. We then fix the message-independent erasure pattern, as in Figure~\ref{fig:fixed-pattern}, and apply the noiseless bound to its surviving transmissions. Averaging over patterns gives the converse in \eqref{eq:intro-main}.

For achievability, we use a fixed-blocklength version of the classical repeat-until-received scheme for erasure channels~\cite[Theorem~7]{PPV11}; see also~\cite{Devassy16}. Alice labels the messages by $m\in\{1,\ldots,M\}$, sets $\ell=m-1$, and writes $\ell$ in base $d$. Starting with the least significant digit, she transmits each digit $j\in\{0,\ldots,d-1\}$ as the computational-basis state $\ket j$. Bob reports whether each transmission was erased. Alice repeats the same digit after an erasure and advances to the next digit after a successful reception. After the prescribed $n$ forward uses, Bob guesses uniformly among the messages consistent with the received digits. We call this the \emph{repeat-until-received protocol}. Sending the least significant digits first means that $k$ successful receptions reveal $\ell\bmod d^k$; Figure~\ref{fig:residues} illustrates the resulting message classes. This protocol attains the converse for every integer $M$. Its forward blocklength is fixed, even though the number of successfully received digits is random. Consequently, quantum encoders and quantum memories do not improve the optimum once classical feedback is available.

We also give two exact extensions of the operational formula. Corollary~\ref{cor:capacity} inverts it to obtain a finite-sum expression for the largest message size at a prescribed error tolerance. Corollary~\ref{cor:correlated} shows that the same patternwise argument applies to correlated erasures, provided the entire erasure pattern is sampled independently of the message.

For codes without feedback, Theorem~\ref{thm:relative} gives a complementary explanation based on relative majorization~\cite{BG17,Renes16}. Every erasure-channel output can be obtained from a pure product reference output by a channel that fixes every state that is maximally mixed within each erasure sector. Theorem~\ref{thm:testing-value} optimizes this entire family of auxiliary states in a hypothesis-testing converse: the optimal weights are proportional to the erasure-pattern probabilities times their sector dimensions, truncated at a threshold. Proposition~\ref{prop:full-meta} shows that this is also the exact value of the fully optimized classical--quantum hypothesis-testing converse. It recovers the same numerical upper bound without feedback, but does not imply operational attainability.

In fact, we prove a strict distinction among classical coding, quantum coding, and feedback-assisted coding at finite blocklength. Theorem~\ref{thm:strict-gap} shows that four messages and four uses of a qubit erasure channel cannot attain \eqref{eq:intro-main} without feedback, even with entangled codewords and a collective decoder. In contrast to the feedback-assisted setting, quantum encoding can strictly improve performance without feedback: we construct a product-state code with three classical carriers and one tetrahedral qubit carrier that outperforms every classical binary erasure code with the same parameters. For $p=1/2$, these results give
\begin{equation}
 \frac{51}{64}
 <\frac{47+2\sqrt6}{64}
 \leq\succp^{*,\nf}(4,4)
 <\frac{53}{64},
 \label{eq:intro-hierarchy}
\end{equation}
where $\succp^{*,\nf}$ denotes the optimum without feedback. The first and last quantities are the optimal classical no-feedback and feedback-assisted values, respectively. The middle explicit quantity is achievable; we do not claim that it is the optimal quantum no-feedback value. Corollary~\ref{cor:capacity-gap} also gives an exact capacity separation: at error $11/64$, the largest message sizes without and with feedback are three and four, respectively.

Finally, Section~\ref{sec:asymptotic} compares asymptotic performance with and without feedback. Corollaries~\ref{cor:normal} and~\ref{cor:normal-nf} give the same normal approximation, with a remainder bounded independently of every positive blocklength, under both average and maximum error. The constant may depend on the fixed channel parameters and error tolerance. Corollary~\ref{cor:exponent} shows that the average-success strong-converse exponent is also the same without feedback, by comparing the exact feedback value with a classical random-code construction. Below capacity, Corollary~\ref{cor:reliability-fb} determines the exact feedback-assisted error exponent. Proposition~\ref{prop:reliability-nf} proves equality with the no-feedback error exponent above a critical rate and gives random-coding, expurgated, and converse bounds at lower rates. We do not infer an exact low-rate no-feedback reliability function from these bounds.

\subsection{Relation to prior work}
Classical feedback has a long history in information theory~\cite{Shannon56}. Quantum channels require additional distinctions among classical feedback, quantum feedback, and initial shared entanglement. Bowen studied quantum feedback and its relation to entanglement-assisted communication~\cite{Bowen04}. Bowen and Nagarajan proved that classical feedback does not increase classical capacity for entanglement-breaking channels, and also established a no-increase result under a restriction to nonentangled encoding schemes~\cite{BN05}. Ref.~\cite{DW18} established a strong converse and exponent bounds for the classical-feedback-assisted capacity of entanglement-breaking channels. Cooney \textit{et al}.~treated strong-converse exponents for quantum-feedback-assisted communication, a setting with a backward quantum channel rather than the backward classical channel allowed here~\cite{CMW16}.

Finite-blocklength bounds for the classical erasure channel are treated by Polyanskiy, Poor, and Verd\'u~\cite{PPV10}. Their erasure converse, stated as Theorem~38 therein, has the same dimension-counting form as \eqref{eq:intro-main}. Polyanskiy, Poor, and Verd\'u also describe transmitting each information bit until it is received in a variable-length feedback scheme~\cite[Theorem~7]{PPV11}. Devassy \textit{et al}.~analyze this scheme and related nonasymptotic bounds for binary erasure channels with feedback~\cite{Devassy16}. We use the same classical retransmission principle with a fixed deadline and a final guess among compatible message labels. Sending least significant digits first provides a convenient residue-class description for arbitrary $M$ and arbitrary alphabet size $d$; we do not claim the retransmission method as a new classical coding principle. The present converse must additionally account for adaptive quantum encoders, retained quantum memories, and entanglement established during the protocol.

Ding \textit{et al}.~proved that classical feedback does not increase the asymptotic classical capacity of a quantum erasure channel~\cite{DQSW19}. Their Theorem~5 gives a weak converse. The strong converse was subsequently established by Ding \textit{et al}.~\cite[Section~VII-D]{DKQSW23}. Accordingly, we do not present the asymptotic capacity or its strong-converse property as new conclusions. Our main result specifies the optimal success probability at every blocklength and message size, and supplies a direct proof tailored to the erasure channel. Wang \textit{et al}.~previously established the second-order asymptotics of the classical capacity of the quantum erasure channel without feedback~\cite[Theorem~13]{WFT19}. For the classical binary erasure channel, the bounded remainder under both average and maximum error is already known from~\cite[Theorem~53]{PPV10}. For the quantum channel, we recover the same first- and second-order terms with and without classical feedback and obtain bounded, rather than logarithmic, remainders. The no-feedback achievability uses classical random coding; the quantum converse permits entangled codewords. The asymptotic comparisons in Section~\ref{sec:asymptotic} also specialize the classical random-coding and expurgation methods~\cite{Gallager65} to the present error criteria.

Rosati studied the $n$-shot classical capacity of the quantum erasure channel without feedback~\cite{Rosati24}. Equation~(34) of that work already gives the numerical converse recovered here, and Eq.~(35) expresses the corresponding message-size bound through truncated sums. We therefore do not claim a tighter no-feedback numerical converse. The relative-majorization argument explains why optimizing an appropriate hypothesis-testing auxiliary state recovers that bound. Our unassisted examples also show why the converse cannot always be interpreted as an exact capacity formula.

Tomamichel \textit{et al}.~obtained finite-resource results for quantum transmission with classical postprocessing, including a closely related erasure-channel auxiliary distribution~\cite{TBR16}. Their discussion following Eq.~(65) concerns a different communication task and a different allocation of free classical communication. Its achievability argument cannot be used as an achievability proof for the classical-message task considered here. We make the auxiliary-state construction explicit in terms of the erasure sectors. The comparison of pairs is related to quantum relative Lorenz curves~\cite{BG17}, relative submajorization~\cite{Renes16}, and the resource theory of asymmetric distinguishability~\cite{WW19}. Here we construct the converting quantum channel explicitly; no sufficiency claim about scalar Lorenz-curve inequalities for arbitrary quantum state pairs is needed.

Dimension bounds for communicating classical information through quantum systems also arise in quantum communication complexity~\cite{NS06}. The restriction to classical backward communication and the absence of initial shared entanglement are essential to the exponent of the dimension in our bound. We explain this restriction in the proof rather than infer it from an entanglement-assisted communication bound.

\subsection{Organization of the paper}
The rest of our paper is organized as follows.
Section~\ref{sec:model} specifies the channel, the resources, and the error criteria. Section~\ref{sec:noiseless} proves the noiseless feedback dimension bound, starting with a three-transmission example. Section~\ref{sec:exact} proves the exact erasure-channel formula and evaluates the maximum message size. Section~\ref{sec:majorization} develops the relative-majorization and hypothesis-testing arguments. Sections~\ref{sec:strict} and~\ref{sec:tetra} establish the finite-blocklength separations. Section~\ref{sec:asymptotic} gives asymptotic consequences, and Section~\ref{sec:conclusion} concludes. Appendix~\ref{app:symmetry} proves the exact value of the fully optimized classical--quantum hypothesis-testing converse, and Appendix~\ref{app:expurgation} derives the classical expurgated error bound used in Section~\ref{sec:asymptotic}.

\section{Notation and communication model}\label{sec:model}
We first specify the resources because distinctions among classical feedback, quantum feedback, and classical postprocessing matter at finite blocklength. We take all quantum systems throughout to be finite dimensional, and we take all logarithms to be binary. For a Hermitian operator $X$, its positive part is denoted by $X_+$, and $(x)_+\coloneqq\max\{x,0\}$ for a real number $x$. A quantum channel is a completely positive and trace-preserving linear map. A quantum instrument is a collection of completely positive, trace-nonincreasing maps whose sum is a channel. A positive operator-valued measure (POVM) is a collection of positive semidefinite operators that sums to the identity.

\subsection{The erasure channel and its sectors}
An erasure differs from an undetected error: the receiver learns which input systems arrived and which were replaced by a flag. For several channel uses, these flags specify a subset of surviving positions. Each such subset corresponds to an orthogonal subspace of the total output space, which we call an erasure sector. A receiver can identify the sector without disturbing the quantum state within it. We first define the single-use channel, then write its tensor-power output as a direct sum over these sectors. The surviving state in a sector is a marginal of the original input, so the description also applies to inputs entangled across different uses. The sector dimensions and probabilities introduced here will enter both the feedback converse in Section~\ref{sec:exact} and the relative-majorization construction in Section~\ref{sec:majorization}.

Let $d\geq2$ and $p\in[0,1]$. The input space of $\cE^p$ has dimension $d$, and the output space is the orthogonal direct sum of the input space and a one-dimensional erasure space. The channel is defined as
\begin{equation}
 \cE^p(X)\coloneqq(1-p)X+p\Tr[X]\proj e,
 \qquad \langle e|a\rangle=0
 \quad\text{for every input vector }\ket a.
 \label{eq:erasure-def}
\end{equation}
Fix a blocklength $n\in\N$. For $S\subseteq[n]\coloneqq\{1,\ldots,n\}$, let $S$ denote the positions that survive. Set
\begin{equation}
 w_S\coloneqq(1-p)^{|S|}p^{n-|S|},\qquad
 D_S\coloneqq d^{|S|},\qquad
 b_k\coloneqq\binom nk(1-p)^kp^{n-k}.
 \label{eq:sector-weights}
\end{equation}
At the endpoints $p=0$ and $p=1$, these expressions have their usual probability-distribution interpretation. The random variable $K\coloneqq|S|$ has probability mass function $(b_k)_{k=0}^n$.

Write $\rho_S$ for the marginal of an input state on the systems indexed by $S$. When $S$ is empty, $\rho_S$ is the scalar $1$. Identifying each erasure sector with the Hilbert space of its surviving systems gives
\begin{equation}
 (\cE^p)^{\otimes n}(\rho_{A_1\cdots A_n})
 =\bigoplus_{S\subseteq[n]}w_S\rho_S.
 \label{eq:direct-sum}
\end{equation}
The erasure flags are implicit in the direct sum. Distinct sectors are orthogonal, and their dimensions sum to $(d+1)^n$. For example, with two uses and the sectors ordered as $\{1,2\},\{1\},\{2\},\emptyset$,
\begin{equation}
 (\cE^p)^{\otimes2}(\rho_{A_1A_2})
 =(1-p)^2\rho_{A_1A_2}\oplus p(1-p)\rho_{A_1}
 \oplus p(1-p)\rho_{A_2}\oplus p^2.
 \label{eq:two-sector-example}
\end{equation}
The last block is one dimensional. This example also shows why different marginals of an entangled input appear in different sectors.

\subsection{Codes, feedback, and error}\label{subsec:codes}
We now specify the two communication models compared in this paper and the error criteria used to evaluate them. A code without feedback chooses all channel inputs before receiving any information from Bob. A feedback-assisted code may adapt later inputs to Bob's earlier measurement outcomes, while retaining quantum memory. We keep the forward blocklength and the absence of initial shared entanglement fixed in both models, so that any difference in performance is due to the permitted feedback.

The message $m$ is uniform on $[M]\coloneqq\{1,\ldots,M\}$. A code without feedback consists of input states $\rho^m_{A_1\cdots A_n}$ and a receiver POVM $(\Lambda_m)_m$. The states may be entangled across all channel inputs. Its average success probability is
\begin{equation}
 \succp\coloneqq\frac1M\sum_{m=1}^{M}
 \Tr\!\left[\Lambda_m(\cE^p)^{\otimes n}(\rho^m_{A_1\cdots A_n})\right].
 \label{eq:code-success}
\end{equation}
The maximum error is the largest error probability over individual messages. An explicit failure outcome is unnecessary: replacing such an outcome by a guess cannot decrease the success probability.

In a feedback-assisted code, Alice and Bob begin without shared entanglement. Alice may retain quantum memory and apply a channel depending on the message and all feedback received so far. Bob may apply arbitrary instruments to his quantum memory and each new channel output, and communicate their classical outcomes to Alice. Bob makes a final decoding measurement. Figure~\ref{fig:noiseless} illustrates this structure for noiseless forward channels; in an erasure-channel protocol, each forward transfer is replaced by one use of $\cE^p$. All feedback between two forward uses can be collected into one classical record; Alice's intermediate local operations are included in her next encoder. The model permits an initial message-independent classical message $F_0$ from Bob to Alice, as in~\cite{DQSW19}. This initial communication can establish a common random seed, but it carries no information about the message.

There are exactly $n$ forward channel uses. Their schedule is fixed independently of the message, so timing is not an additional communication resource. After decoding becomes possible, the protocol may send dummy systems for the remaining uses. There is no free forward classical channel and no backward quantum channel. All classical feedback alphabets and local memories may be arbitrarily large but finite. Supremums over such protocols are understood when an optimum is not known a priori to be attained.

Let $\succp^{*,\fb}(n,M)$ and $\succp^{*,\nf}(n,M)$ denote the optimal average success probabilities with and without feedback, respectively. For an error tolerance $\eps\in[0,1)$, let $M_\fb^*(n,\eps)$ denote the largest feasible message size. The non-asymptotic capacity and the rate are, respectively,
\begin{equation}
 C_{n,\eps}^{\fb}\coloneqq\log_2 M_\fb^*(n,\eps),\qquad
 \frac1n C_{n,\eps}^{\fb}.
 \label{eq:capacity-definition}
\end{equation}
We use analogous notation without feedback. In particular, $\succp^{*,\nf}(n,M)$ is the maximum of \eqref{eq:code-success} over all $M$ input states and all decoding POVMs. We specify explicitly when a claim also holds under the maximum-error criterion.

\section{A dimension bound for noiseless communication with classical feedback}\label{sec:noiseless}
This section contains the main converse argument. The physical protocol is adaptive, but its probabilities can be written using normalized sender states and a single receiver POVM. The sender states depend on the prescribed classical transcript. That dependence is harmless for the dimension estimate; normalization is the property that matters.

\subsection{The elementary bound and the difficulty caused by feedback}
We begin with the state-discrimination dimension bound of Nayak and Salzman~\cite{NS06}; an earlier form appears in~\cite[Theorem~2.4]{Nayak99}. It bounds the success probability directly, rather than bounding a mutual information. For states $(\rho_m)_{m=1}^{M}$ on a $D$-dimensional system and a POVM $(T_m)_m$, the inequality $\rho_m\leq I$ gives
\begin{equation}
 \frac1M\sum_m\Tr[T_m\rho_m]
 \leq\frac1M\sum_m\Tr[T_m]
 =\frac DM.
 \label{eq:elementary-dimension}
\end{equation}
Combining this with the trivial bound $\succp\leq1$ gives $\min\{1,D/M\}$.

In a feedback protocol, however, later transmitted states can depend on earlier receiver outcomes. Conditioning on such outcomes generally changes the message distribution and the sender's quantum memory. Applying \eqref{eq:elementary-dimension} separately to normalized conditional states would not reproduce the original uniform-message expression. We instead keep all outcome probabilities in the receiver's measurement operators and fix the feedback only as classical controls of Alice's trace-preserving channels.

Figure~\ref{fig:noiseless} depicts the noiseless protocol. The circuit layout follows the usual sender--receiver convention for quantum communication protocols; see, for example, the feedback diagrams in~\cite[Chapter~21]{Wilde17}. Here the backward wires are strictly classical, and the parties do not begin with entanglement.

\begin{figure}[t]
 \centering
 \input{figures/noiseless_protocol.tex}
 \caption{Three noiseless forward transmissions with classical feedback. Solid wires carry quantum systems; dashed wires carry classical data. Alice's retained memories are $A_i$, Bob's memories are $B_i$, and $Q_i$ is the system transmitted in use $i$. The label $m$ denotes Alice's classical control; its repeated appearance does not represent additional communication. Bob's instruments may establish correlations with Alice's retained systems. Their classical outcomes control subsequent encoders. An initial message-independent seed, when present, is suppressed.}
 \label{fig:noiseless}
\end{figure}
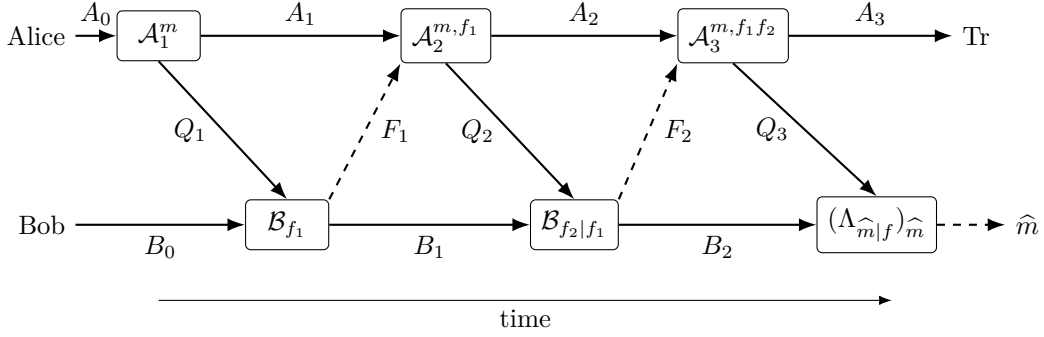

\subsection{Three transmissions: the complete construction}\label{subsec:three-use}
We now write out the three-transmission protocol of Figure~\ref{fig:noiseless}. It contains two successive feedback outcomes, so the second receiver instrument can depend on the first outcome, and the third encoder can depend on both. This example makes explicit how these dependencies are retained when the protocol is represented by a normalized sender state and a single receiver POVM. As throughout the converse, prescribing an outcome as a classical control is different from conditioning on that outcome.

After Alice's first encoding, let $\xi^m_{A_1Q_1}$ be the normalized joint state of her retained memory and the first transmitted system. Incorporate Bob's initial local state into his first instrument $(\cB_{f_1})_{f_1}$, and write his second instrument as $(\cB_{f_2|f_1})_{f_2}$. Their input and output systems are
\begin{equation}
 \cB_{f_1}\colon Q_1\longrightarrow B_1,
 \qquad
 \cB_{f_2|f_1}\colon B_1Q_2\longrightarrow B_2.
 \label{eq:three-bob-instruments}
\end{equation}
Here $\sum_{f_1}\cB_{f_1}$ is trace preserving, and $\sum_{f_2}\cB_{f_2|f_1}$ is trace preserving for each fixed $f_1$. Bob sends $f_1$ after the first transmission and $f_2$ after the second. Alice's subsequent encoder channels are
\begin{equation}
 \cA_2^{m,f_1}\colon A_1\longrightarrow A_2Q_2,
 \qquad
 \cA_3^{m,f_1,f_2}\colon A_2\longrightarrow A_3Q_3.
 \label{eq:three-alice-channels}
\end{equation}
After the third transmission, Bob uses a POVM $(\Lambda_{\widehat m|f})_{\widehat m}$ on $B_2Q_3$, where $f=(f_1,f_2)$ denotes the complete feedback transcript. We also write $\cA_3^{m,f}$ for $\cA_3^{m,f_1,f_2}$.

\paragraph{A normalized state with prescribed controls.}
Fix $m$ and $f=(f_1,f_2)$. Apply Alice's second encoder with its control set to $f_1$, without applying Bob's first instrument. Then apply her third encoder with its controls set to $f_1,f_2$, without applying Bob's second instrument. Retain the three forward systems in this mathematical construction and set
\begin{align}
 \eta^{m,f_1}_{A_2Q_1Q_2}
 &\coloneqq(\cA_2^{m,f_1}\otimes\id_{Q_1})(\xi^m_{A_1Q_1}),
 \label{eq:forced-three-intermediate}\\
 \rho^{m,f}_{Q_1Q_2Q_3}
 &\coloneqq\Tr_{A_3}\!\left[
 (\cA_3^{m,f}\otimes\id_{Q_1Q_2})(\eta^{m,f_1}_{A_2Q_1Q_2})\right].
 \label{eq:forced-three-state}
\end{align}
Tensor factors are reordered as indicated by the system labels after each operation. Both encoders are trace preserving, so
\begin{equation}
 \rho^{m,f}_{Q_1Q_2Q_3}\geq0,
 \qquad\Tr[\rho^{m,f}_{Q_1Q_2Q_3}]=1,
 \qquad\rho^{m,f}_{Q_1Q_2Q_3}\leq I_{Q_1Q_2Q_3}.
 \label{eq:forced-normalization}
\end{equation}
In particular, this construction is defined even when the transcript $f$ has zero probability in the actual protocol. Neither of Bob's outcomes has been postselected. The state need not be a product across the three forward systems.

\paragraph{One POVM containing the complete transcript.}
For the same prescribed $f$, Bob's two receiver operations compose to the completely positive, trace-nonincreasing linear map
\begin{equation}
 \mathcal D_f\coloneqq
 \cB_{f_2|f_1}\circ(\cB_{f_1}\otimes\id_{Q_2})
 \colon Q_1Q_2\longrightarrow B_2.
 \label{eq:three-bob-composition}
\end{equation}
The sum $\sum_f\mathcal D_f$ is a channel. This follows by first summing over $f_2$ for each fixed $f_1$, and then summing over $f_1$. Pull the final decoding effect back through these two operations:
\begin{equation}
 \begin{aligned}
 T_{f,\widehat m}
 &\coloneqq(\mathcal D_f^\dagger\otimes\id_{Q_3})(\Lambda_{\widehat m|f})\\
 &= (\cB_{f_1}^\dagger\otimes\id_{Q_2Q_3})
 \!\left[(\cB_{f_2|f_1}^\dagger\otimes\id_{Q_3})(\Lambda_{\widehat m|f})\right].
 \end{aligned}
 \label{eq:three-effects}
\end{equation}
The adjoints are taken with respect to the trace pairing. Complete positivity gives $T_{f,\widehat m}\geq0$. To check completeness, sum the final estimate first, the second feedback outcome next, and the first outcome last:
\begin{align}
 \sum_{f,\widehat m}T_{f,\widehat m}
 &=\sum_{f_1}(\cB_{f_1}^\dagger\otimes\id_{Q_2Q_3})
 \!\left[\sum_{f_2}(\cB_{f_2|f_1}^\dagger\otimes\id_{Q_3})(I_{B_2Q_3})\right]
 \label{eq:three-completeness-a}\\
 &=\sum_{f_1}(\cB_{f_1}^\dagger\otimes\id_{Q_2Q_3})(I_{B_1Q_2Q_3})
 \label{eq:three-completeness-b}\\
 &=\left(\sum_{f_1}\cB_{f_1}^\dagger(I_{B_1})\right)\otimes I_{Q_2Q_3}
 =I_{Q_1Q_2Q_3}.
 \label{eq:three-completeness-c}
\end{align}
Each summation uses unitality of the adjoint of a trace-preserving channel. Thus $(T_{f,\widehat m})_{f,\widehat m}$ is one POVM on all three forward systems, not a separate normalized POVM for each transcript.

\paragraph{Recovering the actual probabilities.}
Before the final measurement, the unnormalized state for a fixed transcript in the actual adaptive protocol is
\begin{equation}
 \chi^{m,f}_{A_3B_2Q_3}
 \coloneqq(\cA_3^{m,f}\otimes\cB_{f_2|f_1})
 \!\left[(\cA_2^{m,f_1}\otimes\cB_{f_1})(\xi^m_{A_1Q_1})\right].
 \label{eq:three-actual-state}
\end{equation}
In this expression Alice's third encoder acts on $A_2$, whereas Bob's second operation acts on $B_1Q_2$. These operations commute on their disjoint systems. Similarly, Alice's second encoder acts on $A_1$, whereas Bob's first operation acts on $Q_1$. Moving the Alice operations before the Bob operations therefore gives
\begin{equation}
 \Tr_{A_3}[\chi^{m,f}_{A_3B_2Q_3}]
 =(\mathcal D_f\otimes\id_{Q_3})(\rho^{m,f}_{Q_1Q_2Q_3}).
 \label{eq:three-reordering}
\end{equation}
This uses commutation of local operations, not a separability assumption: the memories and transmitted systems can be entangled throughout. The classical controls have already been prescribed separately for each $f$; no operation is commuted through an unfixed feedback control.

Using \eqref{eq:three-reordering} and the defining property of the adjoint, the actual joint probability is
\begin{align}
 \Pr\!\left(f,\widehat m|m\right)
 &=\Tr\!\left[\Lambda_{\widehat m|f}\Tr_{A_3}[\chi^{m,f}_{A_3B_2Q_3}]\right]
 \label{eq:three-probability-a}\\
 &=\Tr\!\left[\Lambda_{\widehat m|f}
 (\mathcal D_f\otimes\id_{Q_3})(\rho^{m,f}_{Q_1Q_2Q_3})\right]
 \label{eq:three-probability-b}\\
 &=\Tr[T_{f,\widehat m}\rho^{m,f}_{Q_1Q_2Q_3}].
 \label{eq:three-probability-c}
\end{align}
All probability factors associated with $f_1$ and $f_2$ remain in the effect $T_{f,\widehat m}$. We have not divided by $\Pr(f|m)$ or assumed that the message remains uniform after observing a transcript.

Figure~\ref{fig:fixed-transcript} summarizes the resulting representation. Different transcripts generally give different normalized states $\rho^{m,f}$. Nevertheless, each state is bounded above by the same identity, and all the effects sum to that identity. Consequently,
\begin{equation}
 \begin{aligned}
 \succp
 &=\frac1M\sum_{m,f}\Tr[T_{f,m}\rho^{m,f}_{Q_1Q_2Q_3}]\\
 &\leq\frac1M\sum_{m,f}\Tr[T_{f,m}]
 =\frac{\Tr[I_{Q_1Q_2Q_3}]}{M}
 =\frac{\dim Q_1\dim Q_2\dim Q_3}{M}.
 \end{aligned}
 \label{eq:three-dimension}
\end{equation}
The sum over $m$ in the middle expression is precisely the sum over all final estimates in the complete POVM. No simultaneous physical preparation of the transcript-dependent states is required for this inequality.

\begin{figure}[t]
 \centering
 \input{figures/fixed_transcript.tex}
 \caption{The fixed-control representation of the three-transmission protocol in Figure~\ref{fig:noiseless}, with $f=(f_1,f_2)$. For each $m,f$, Alice's prescribed trace-preserving encoders prepare a normalized state. Bob's complete sequence of instruments defines one POVM, whose outcomes include both feedback values and the final estimate. Commuting operations on disjoint systems gives equality of the trace pairings. The diagram represents individual joint probabilities; it does not assert that one nonadaptive encoder physically implements all transcripts at once.}
 \label{fig:fixed-transcript}
\end{figure}
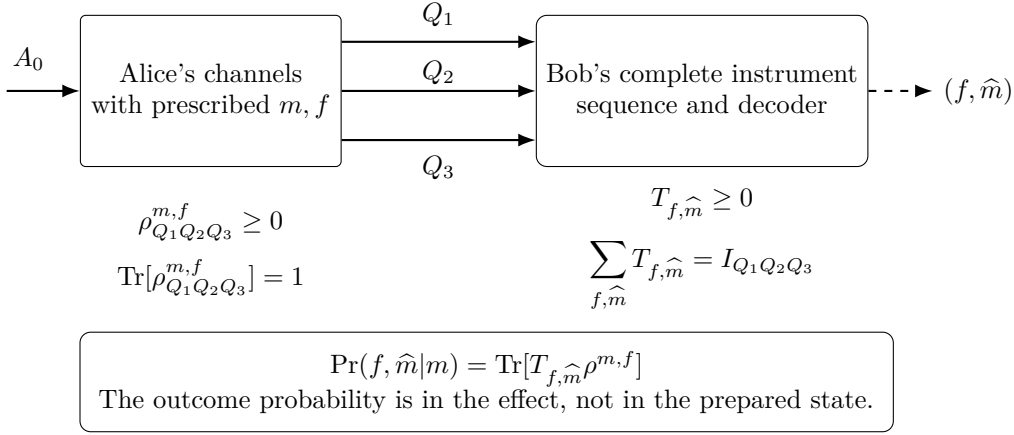

\subsection{An arbitrary number of transmissions}
The three-transmission argument extends without changing the dimension estimate. The proof below explains explicitly how the state preparation and the POVM are defined for the complete transcript.

\begin{samepage}
\begin{lemma}[Noiseless quantum communication with classical feedback]\label{lem:noiseless}
Suppose Alice and Bob initially share no entanglement. Alice transmits systems $Q_1,\ldots,Q_k$ through noiseless quantum channels, with $\dim Q_i=d_i$, and Bob may send arbitrary classical feedback. The protocol may retain arbitrary local quantum memories and generate entanglement during its execution. For $M$ equiprobable messages,
\begin{equation}
 \succp\leq\min\!\left\{1,\frac{\prod_{i=1}^k d_i}{M}\right\}.
 \label{eq:noiseless-bound}
\end{equation}
The empty product is $1$. The conclusion also holds when a message-independent initial feedback seed is supplied.
\end{lemma}
\end{samepage}
\begin{proof}
If $k=0$, Bob receives no forward system, so all of his data are independent of the message and his average success probability is $1/M$. Thus assume $k\geq1$.

First fix an initial seed, if present. The message is still uniform, and the initial local states are a product for this fixed seed. More generally, an initially shared separable state can be handled by revealing to both parties a label in a separable decomposition. This only enlarges their resources and gives product states conditioned on a message-independent label. Bob's local initial state can therefore be incorporated into his first instrument.

Let $f=(f_1,\ldots,f_{k-1})$ denote the feedback transcript returned between successive forward transmissions; any receiver outcomes not reported to Alice can remain in Bob's memory. For a fixed $m,f$, apply Alice's encoder channels in order, using the prescribed transcript prefixes as their classical controls. Retain all transmitted systems $Q_1,\ldots,Q_k$ in this mathematical construction and trace out Alice's final memory. The resulting operator $\rho^{m,f}_{Q_1\cdots Q_k}$ is a state, because it is prepared by a composition of trace-preserving channels followed by a partial trace. In particular,
\begin{equation}
 0\leq\rho^{m,f}_{Q_1\cdots Q_k}\leq I_{Q_1\cdots Q_k}.
 \label{eq:general-normalized}
\end{equation}
Alice's private measurement records, if any, are retained as local classical memory in these channels. They are not conditioned on selected outcomes.

Next regard Bob as being given all systems $Q_1,\ldots,Q_k$ at the beginning, while he still processes them in the original order and leaves future systems untouched. This defines a legitimate measurement procedure on an arbitrary joint state of the $Q_i$. Its outcomes are $(f,\widehat m)$. Denote its POVM by $(T_{f,\widehat m})_{f,\widehat m}$. Positivity follows from complete positivity of the constituent instruments. Summing the last outcome replaces the last measurement by the identity. Summing the preceding instrument outcome then replaces the adjoint of that instrument by a unital channel. Continuing backward through the instrument tree gives
\begin{equation}
 T_{f,\widehat m}\geq0,
 \qquad\sum_{f,\widehat m}T_{f,\widehat m}=I_{Q_1\cdots Q_k}.
 \label{eq:general-completeness}
\end{equation}
This is the same completeness argument as \eqref{eq:three-completeness-a}--\eqref{eq:three-completeness-c}, iterated over the tree.

For a fixed transcript, Alice's later encoder acts on her retained memory, whereas each earlier receiver operation acts on systems already transmitted. These local maps commute, even on an entangled joint input. Move each later Alice map to the left of the earlier Bob maps, one at a time. This reordering changes neither the unnormalized state for that transcript nor its trace. It yields
\begin{equation}
 \Pr\!\left(f,\widehat m| m\right)
 =\Tr[T_{f,\widehat m}\rho^{m,f}_{Q_1\cdots Q_k}].
 \label{eq:general-trace-pairing}
\end{equation}
No commutation through a feedback control is asserted: the classical controls have already been fixed, separately for each transcript. The expression on the right still contains the probability of that transcript.

Using \eqref{eq:general-normalized} and \eqref{eq:general-completeness}, we conclude that
\begin{align}
 \succp
 &=\frac1M\sum_{m,f}\Tr[T_{f,m}\rho^{m,f}_{Q_1\cdots Q_k}]
 \label{eq:general-bound-a}\\
 &\leq\frac1M\sum_{m,f}\Tr[T_{f,m}]
 =\frac{\Tr[I_{Q_1\cdots Q_k}]}{M}
 =\frac{\prod_i d_i}{M}.
 \label{eq:general-bound-b}
\end{align}
Averaging over the initial seed preserves the bound. The probability is also at most one, proving \eqref{eq:noiseless-bound}.
\end{proof}

\begin{remark}[Why the resource assumptions matter]\label{rem:resources}
The proof allows the parties to establish entanglement by earlier forward transmissions. It does not require intermediate states to be separable. It does require that no shared entanglement be available initially and that the backward communication be classical. With an initially shared maximally entangled state, dense coding permits one transmitted $d$-dimensional system to convey $d^2$ perfectly distinguishable messages~\cite{Wilde17}. A backward quantum channel can also establish such entanglement for free. These resources invalidate \eqref{eq:noiseless-bound}; they are not included in our model. In contrast, Alice may use her first noiseless transmission to distribute one share of a maximally entangled state and then use her second transmission for dense coding. This conveys $d^2$ messages using two $d$-dimensional transfers, exactly consistent with the dimension product $d^2$ in Lemma~\ref{lem:noiseless}.
\end{remark}

\section{Exact feedback-assisted performance of the erasure channel}\label{sec:exact}
The noiseless bound depends only on the dimensions of the quantum systems that cross from Alice to Bob. For an erasure channel, a fixed pattern identifies exactly which systems do so. We first apply this observation to arbitrary quantum protocols assisted by classical feedback, then exhibit an elementary protocol that attains the resulting bound.

\subsection{Conditioning on the erasure pattern}
We first state the exact performance formula and prove its converse. The key distinction from a general measurement outcome is that the erasure pattern is sampled independently of the message, even when the encoder is adaptive. Fixing that pattern therefore leaves an ordinary uniform-message communication problem, with one noiseless transfer for each surviving input. Figure~\ref{fig:fixed-pattern} illustrates the reduction to Lemma~\ref{lem:noiseless}. The following subsection supplies the matching achievability protocol.

To average the patternwise dimension bound, define the success function
\begin{equation}
 s_{n,d,p}(M)\coloneqq\sum_{k=0}^n b_k\min\!\left\{1,\frac{d^k}{M}\right\},
 \qquad M\geq1.
 \label{eq:success-function}
\end{equation}
Here $b_k=\binom nk(1-p)^kp^{n-k}$ is the probability of $k$ survivors, as defined in \eqref{eq:sector-weights}. Although message sizes are integers, allowing real $M$ will be useful when we invert this function.

\begin{theorem}[Exact fixed-blocklength feedback formula]\label{thm:feedback}
For $n\in\N$, $d\geq2$, $p\in[0,1]$, and an integer $M\geq1$,
\begin{equation}
 \succp^{*,\fb}(n,M)=s_{n,d,p}(M).
 \label{eq:feedback-exact}
\end{equation}
The converse holds for arbitrary adaptive quantum protocols in the model of Section~\ref{sec:model}. The achievability uses only orthogonal input states and erasure acknowledgments. If an initial message-independent feedback seed is permitted, the optimal maximum error equals the optimal average error, both given by
\begin{equation}
 \eps_\fb^*(n,M)
 =\sum_{k=0}^n b_k\left(1-\frac{d^k}{M}\right)_+.
 \label{eq:error-exact}
\end{equation}
\end{theorem}
\begin{proof}[Proof of the converse]
Each channel use is a mixture of two trace-preserving channels: a noiseless transmission and replacement by a fixed flag. Consequently, the entire erasure pattern can be sampled before the protocol begins, with probability $w_S$. This remains true for an adaptive input. At each use, the choice between the two trace-preserving components has probabilities $1-p$ and $p$ independently of the input and of the previous history. In particular, the sampled pattern is independent of the message.

Fix $S$. For the converse, this pattern may be fixed in advance in a comparison protocol, or revealed to both parties: it is independent of the message. In this fixed-pattern protocol, an erased transmission can be replaced by Alice discarding that input locally and Bob receiving a locally prepared flag $\proj e$. For a joint state of the input and arbitrary retained systems, both operations take the partial trace over the erased input. Thus they have the same effect on every accessible system, including entangled memories. The receiver's later operations and all feedback are left unchanged.

Only the $|S|$ surviving systems now cross from Alice to Bob, each through a noiseless $d$-dimensional channel. The erased positions merely insert local operations and classical feedback between these transfers. Applying Lemma~\ref{lem:noiseless} therefore gives
\begin{equation}
 \Pr\!\left(\widehat m=m| S\right)
 \leq\min\!\left\{1,\frac{d^{|S|}}{M}\right\},
 \label{eq:pattern-converse}
\end{equation}
where the success probability is averaged over the still-uniform message. This conditioning is different from conditioning on a receiver measurement outcome: $S$ is an external random pattern independent of the message. Figure~\ref{fig:fixed-pattern} depicts the reduction.

Average \eqref{eq:pattern-converse} over $S$. The sum of $w_S$ over patterns of size $k$ is $b_k$, so the result is $\succp\leq s_{n,d,p}(M)$. A maximum-error code is also an average-error code with the same error upper bound, so the converse applies under either criterion.
\end{proof}

\begin{figure}[t]
 \centering
 \input{figures/fixed_pattern.tex}
 \caption{The fixed-pattern reduction for four channel uses with survivors $S=\{1,3\}$. After fixing this message-independent pattern, an erased input is traced out on Alice's side and an erasure flag is prepared on Bob's side. Feedback and memory updates remain unrestricted. The apparent gap between a discard and a flag denotes local replacement, not a quantum wire. The resulting protocol has only two noiseless forward quantum transfers, so Lemma~\ref{lem:noiseless} applies.}
 \label{fig:fixed-pattern}
\end{figure}
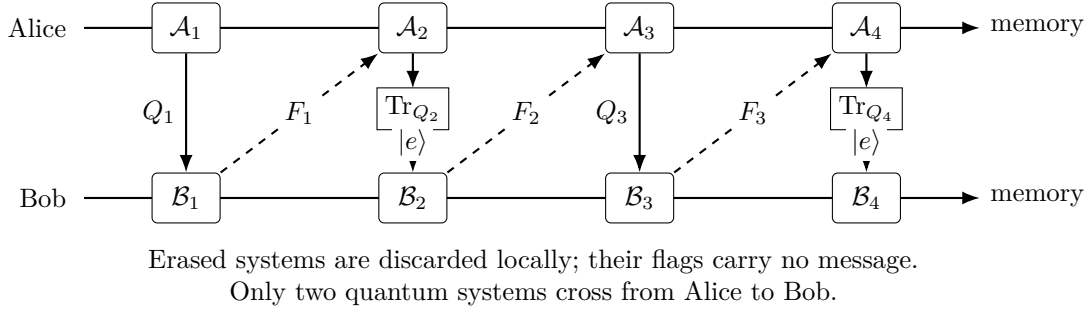

\subsection{A protocol that attains the bound}
The converse counts how many distinguishable possibilities $k$ successful transmissions could convey. The classical repeat-until-received scheme sends each symbol again until the receiver acknowledges it~\cite[Theorem~7]{PPV11}; see also~\cite{Devassy16}. In those references, stopping after all information symbols have arrived gives a variable-length protocol. Here we stop after exactly $n$ forward uses, whether or not the message has been identified, and then make a final guess. The following choice of digit order attains the dimension bound simultaneously for every $k$, even when $M$ is not a power of $d$.

\begin{proof}[Proof of achievability in Theorem~\ref{thm:feedback}]
For a message $m\in[M]$, set $\ell\coloneqq m-1\in\{0,\ldots,M-1\}$. Alice writes this auxiliary integer $\ell$ in base $d$ and sends its least significant digit as a computational-basis state. The message labels remain $1,\ldots,M$; the subtraction by one is only a convenient representation for the encoding. Sending the least significant digits first ensures that $k$ successful receptions reveal $\ell\bmod d^k$. If Bob receives an erasure flag, he reports the erasure, and Alice repeats the same digit. Once the digit arrives, Bob acknowledges it, and Alice advances to the next digit. The procedure continues until $n$ forward uses have taken place. After all necessary digits have arrived, Alice sends leading zero digits, which serve as dummy states. This convention defines the received residue for every possible number of successful transmissions.

After exactly $k$ successful transmissions, Bob knows
\begin{equation}
 \ell\bmod d^k=(m-1)\bmod d^k.
 \label{eq:residue-known}
\end{equation}
Among the auxiliary integers $\ell=0,\ldots,M-1$, exactly $\min\{M,d^k\}$ residue classes modulo $d^k$ are nonempty. Conditional on the observed class of auxiliary integers, Bob guesses uniformly among the corresponding messages $m=\ell+1$. A nonempty class of size $a$ contains $a$ messages, each contributing $1/(Ma)$ to the conditional average success probability. Hence every nonempty class contributes $1/M$, and
\begin{equation}
 \Pr\!\left(\widehat m=m| K=k\right)
 =\frac{\min\{M,d^k\}}{M}
 =\min\!\left\{1,\frac{d^k}{M}\right\}.
 \label{eq:digit-success}
\end{equation}
The erasure pattern and therefore $K$ are independent of the message. Averaging proves \eqref{eq:feedback-exact}. Figure~\ref{fig:residues} illustrates the residue classes for seven messages and binary digits.

For maximum error, Bob first chooses $R$ uniformly in $\{0,\ldots,M-1\}$ and sends $R$ to Alice over the free backward classical channel. With $\ell=m-1$ as above, they use the shifted encoding integer
\begin{equation}
 \ell'\coloneqq(\ell+R)\bmod M,
 \qquad \ell'=m'-1,
 \label{eq:cyclic-seed}
\end{equation}
where $m'=\ell'+1$ is the corresponding message label. They run the protocol on the digits of $\ell'$. Bob converts his estimate back by $\widehat\ell=(\widehat\ell'-R)\bmod M$ and $\widehat m=\widehat\ell+1$. For each fixed original message $m$, the shifted integer $\ell'$ is uniform when averaged over the seed. The error for every $m$ is therefore the average error of the original protocol. No forward communication has been added. This proves the maximum-error assertion and \eqref{eq:error-exact}.
\end{proof}

\begin{figure}[t]
 \centering
 \input{figures/residue_tree.tex}
 \caption{The repeat-until-received protocol for $M=7$ and $d=2$. Nodes contain the auxiliary integers $\ell=m-1\in\{0,\ldots,6\}$, whose binary digits are transmitted starting with the least significant digit. The corresponding message label is always $m=\ell+1$. With $k=0,1,2$ successful receptions, the numbers of nonempty sets are $1,2,4$, respectively; with three successful receptions all seven messages are identified. A receiver erasure acknowledgment leaves the current set unchanged. The protocol has fixed blocklength: the number of successful steps is random, not the number of forward uses.}
 \label{fig:residues}
\end{figure}

\begin{remark}[Initial feedback and maximum error]\label{rem:initial-seed}
Without the initial seed, the unsymmetrized repeat-until-received protocol need not have equal individual-message errors when $M$ is not a power of $d$. We do not silently identify average and maximum error in that restricted model. The average-error formula requires no initial seed, while the maximum-error formula stated here uses the explicitly permitted initial classical feedback.
\end{remark}

The converse and achievability above both hold separately for each erasure pattern. Independence among different erasure events was used only to write their probabilities in binomial form. This observation gives the following extension.

\begin{samepage}
\begin{corollary}[Message-independent correlated erasures]\label{cor:correlated}
Suppose the entire survivor set $S$ is drawn according to an arbitrary distribution $(w_S)_S$, independently of the message, and the conditional forward channel is noiseless on $S$ and erased on its complement. For equal input dimensions $d$, the optimal feedback-assisted average success probability is
\begin{equation}
 \sum_{S\subseteq[n]}w_S\min\!\left\{1,\frac{d^{|S|}}{M}\right\}.
 \label{eq:correlated-erasure}
\end{equation}
The same value holds for maximum error with an initial seed.
\end{corollary}
\end{samepage}
\begin{proof}
The fixed-pattern converse uses only independence of the pattern from the message, not independence among its coordinates. The repeat-until-received protocol has the conditional success probability in \eqref{eq:digit-success} for each pattern separately. Its seed symmetrization also does not use independence among erasures.
\end{proof}

\subsection{An explicit maximum-message formula}\label{subsec:inverse}
Theorem~\ref{thm:feedback} gives the optimal decoding error for each integer message size. We now invert that relation: for a prescribed error tolerance, we seek the largest number of messages that can be transmitted. It is convenient to solve the inequality first for a real message-size parameter and then return to integer message sizes. This also yields the cumulative quantities used to construct an optimal hypothesis-testing auxiliary state in Section~\ref{sec:majorization}.

With $n$, $d$, and $p$ fixed, define
\begin{equation}
 e(r)\coloneqq\sum_{k=0}^n b_k\left(1-\frac{d^k}{r}\right)_+,
 \qquad r\geq1.
 \label{eq:real-error}
\end{equation}
For a positive integer $M$, the coding theorem identifies this function with the optimal feedback-assisted average error:
\begin{equation}
 e(M)=\eps_{\fb}^*(n,M)=1-\succp^{*,\fb}(n,M).
 \label{eq:real-error-operational}
\end{equation}
Under the initial-feedback-seed convention of Theorem~\ref{thm:feedback}, it is also the optimal maximum error. For a noninteger argument, $e(r)$ is only a continuous, nondecreasing extension of this error function; it does not describe a code with a fractional number of messages. The parameter $r$ represents a message size, not a communication rate.

For an integer message size $M$, the attaining protocol has conditional error $(1-d^k/M)_+$ when $k$ transmissions survive. This error is positive only if the surviving dimension $d^k$ is smaller than $M$. The same sign pattern holds in \eqref{eq:real-error} for a real $r$: since $d^k$ increases with $k$, the factors $1-d^k/r$ are positive only for an initial segment of survivor counts. Once a factor becomes nonpositive, every subsequent factor is nonpositive as well. We can therefore replace the individual positive parts by a suitable cutoff in the survivor count.

For each cutoff $j$, set
\begin{equation}
 B_j\coloneqq\sum_{k=0}^j b_k,\qquad
 Z_j\coloneqq\sum_{k=0}^j b_kd^k,\qquad 0\leq j\leq n.
 \label{eq:cumulative-defs}
\end{equation}
Here $B_j=\Pr(K\leq j)$ is the cumulative probability of at most $j$ survivors, while $Z_j$ is the corresponding sum weighted by the surviving dimensions. For $\eps\in[0,1)$, define
\begin{equation}
 R_\eps\coloneqq
 \min\!\left\{\frac{Z_j}{B_j-\eps}:0\leq j\leq n,\ B_j>\eps\right\}.
 \label{eq:R-definition}
\end{equation}
The feasible set is nonempty because $B_n=1$. The next corollary identifies $R_\eps$ as the largest admissible real message-size parameter and then rounds it down.

\begin{corollary}[Exact feedback-assisted capacity]\label{cor:capacity}
Under the assumptions of Theorem~\ref{thm:feedback},
\begin{equation}
 M_\fb^*(n,\eps)=\lfloor R_\eps\rfloor,
 \qquad C_{n,\eps}^{\fb}=\log_2\lfloor R_\eps\rfloor.
 \label{eq:capacity-formula}
\end{equation}
These formulas apply to average error and to maximum error with an initial feedback seed.
\end{corollary}
\begin{proof}
To express \eqref{eq:real-error} in terms of $B_j$ and $Z_j$, consider the running sum of the signed contributions, without taking their positive parts:
\begin{equation}
 F_j(r)\coloneqq\sum_{k=0}^j b_k\left(1-\frac{d^k}{r}\right)
 =B_j-\frac{Z_j}{r},
 \qquad 0\leq j\leq n.
 \label{eq:prefix-error-sum}
\end{equation}
Also set $F_{-1}(r)\coloneqq0$, corresponding to the empty sum. Increasing the cutoff from $j-1$ to $j$ changes the running sum by
\begin{equation}
 F_j(r)-F_{j-1}(r)
 =b_j\left(1-\frac{d^j}{r}\right).
 \label{eq:prefix-error-increment}
\end{equation}
Because $b_j\geq0$, this running sum cannot decrease while $d^j<r$, and it cannot increase after $d^j\geq r$. Its largest value is obtained by including all positive contributions and stopping before any negative contribution is added. Terms with $d^j=r$ or $b_j=0$ leave the sum unchanged.

More explicitly, each cutoff $j\geq0$ satisfies
\begin{equation}
 F_j(r)
 \leq\sum_{k=0}^j b_k\left(1-\frac{d^k}{r}\right)_+
 \leq e(r).
 \label{eq:prefix-error-upper}
\end{equation}
Conversely, if at least one index satisfies $d^k<r$, choose $j$ to be the largest such index in $\{0,\ldots,n\}$. All retained factors are positive and all omitted factors are nonpositive, so $F_j(r)=e(r)$. If there is no such index, then $e(r)=0=F_{-1}(r)$. This proves
\begin{equation}
 e(r)=\max\left\{0,\max_{0\leq j\leq n}
 \left(B_j-\frac{Z_j}{r}\right)\right\}.
 \label{eq:error-max-prefix}
\end{equation}
Thus, taking the positive part of each summand is equivalent to choosing the best cutoff. The additional zero represents the option of retaining no terms. The argument also covers the endpoint channels $p=0$ and $p=1$, for which some weights $b_k$ vanish.

Equation~\eqref{eq:error-max-prefix} turns the error constraint into one inequality for each cutoff. Since $\eps\geq0$, the condition $e(r)\leq\eps$ is equivalent to $B_j-Z_j/r\leq\eps$ for every $j$. If $B_j\leq\eps$, this inequality holds automatically because $Z_j\geq0$. If $B_j>\eps$, rearranging it gives
\begin{equation}
 r\leq\frac{Z_j}{B_j-\eps}
 \quad\text{for every }j\text{ satisfying }B_j>\eps.
 \label{eq:prefix-constraints}
\end{equation}
The largest feasible real parameter is therefore the smallest of these upper bounds, namely $R_\eps$. The feasible interval contains $1$ because $e(1)=0$. Theorem~\ref{thm:feedback} attains the error $e(M)$ at each integer $M$, so the largest feasible integer is $\lfloor R_\eps\rfloor$. Taking its binary logarithm proves the capacity formula.
\end{proof}

For example, for $d=2$, $p=1/2$, and $\eps=0.1$, the formula gives $M_\fb^*(20,0.1)=284$, or approximately $8.150$ transmitted bits. Figure~\ref{fig:normal-rate} in Section~\ref{sec:asymptotic} plots the same exact capacity after division by $n$ and compares it with the normal approximation. Before considering those asymptotic consequences, we note that the formula also includes the following endpoint cases.

\begin{example}[Endpoints and message sizes larger than the input dimension]\label{ex:endpoints}
For a noiseless channel, $p=0$, and for a completely erased channel, $p=1$, respectively,
\begin{equation}
 R_\eps=\frac{d^n}{1-\eps}
 \quad(p=0),\qquad
 R_\eps=\frac1{1-\eps}
 \quad(p=1).
 \label{eq:endpoints}
\end{equation}
There is no assumption that $M\leq d^n$: at sufficiently large allowed error, more messages can be accommodated by assigning several messages the same transmitted word. For zero error and $p>0$, the all-erased pattern has positive probability and forces $M_\fb^*(n,0)=1$. These cases are included in~\eqref{eq:capacity-formula}.
\end{example}

\section{Relative majorization and an exactly evaluable converse}\label{sec:majorization}
The feedback theorem was proved without hypothesis testing. We now give a second explanation of the same bound for codes without feedback. The erasure-sector decomposition yields a common family of invariant auxiliary states, and the associated hypothesis-testing optimization reduces to scalar variables. This structural result explains the quantum counterpart of the auxiliary distributions used in erasure-channel converses.

\subsection{A simultaneous invariant-family construction}
The physical erasure probabilities do not have to be used in the auxiliary hypothesis of a converse. We therefore allow arbitrary weights on the erasure sectors while keeping the auxiliary state maximally mixed inside each sector. We will show that one channel converts the reference output into an arbitrary output and simultaneously preserves every member of this auxiliary family; Figure~\ref{fig:sector-conversion} illustrates the two requirements.

For a probability distribution $q=(q_S)_{S\subseteq[n]}$, define
\begin{equation}
 \sigma_q\coloneqq\bigoplus_{S\subseteq[n]}q_S\frac{I_S}{D_S},
 \label{eq:aux-family}
\end{equation}
where $I_S$ is the identity on the surviving-system space in sector $S$. In particular, the choice $q_S=w_S$ gives
\begin{equation}
 \sigma_w=(\cE^p)^{\otimes n}(\pi^{\otimes n}),
 \qquad\pi\coloneqq I/d.
 \label{eq:physical-reference}
\end{equation}
Allowing $q$ to differ from the physical pattern distribution will be essential to the optimized converse.

Given two pairs of states $(\omega,\sigma)$ and $(\omega',\sigma')$, we write
\begin{equation}
 (\omega,\sigma)\longrightarrow(\omega',\sigma')
 \label{eq:relative-order-definition}
\end{equation}
when there is a single quantum channel taking $\omega$ to $\omega'$ and $\sigma$ to $\sigma'$. This common-channel comparison of pairs is closely related to quantum relative Lorenz curves~\cite{BG17}, relative submajorization~\cite{Renes16}, and the resource theory of asymmetric distinguishability~\cite{WW19}. We use \emph{relative majorization} for the concrete common-channel conversion proved next. In a general noncommuting setting, scalar testing or Lorenz-curve inequalities alone need not characterize such a conversion; here the required channel is constructed explicitly.

\begin{theorem}[Relative majorization of erasure-channel outputs]\label{thm:relative}
Fix a pure state $\psi$ on the $d$-dimensional input, and set
\begin{equation}
 \omega_\psi\coloneqq(\cE^p)^{\otimes n}(\psi^{\otimes n}).
 \label{eq:pure-reference}
\end{equation}
For every input state $\rho_{A_1\cdots A_n}$, there is a channel $\cM^\rho$ such that
\begin{equation}
 \cM^\rho(\omega_\psi)=(\cE^p)^{\otimes n}(\rho_{A_1\cdots A_n}),
 \qquad \cM^\rho(\sigma_q)=\sigma_q\quad\text{for every }q.
 \label{eq:simultaneous-conversion}
\end{equation}
Thus the same channel implements
\begin{equation}
 (\omega_\psi,\sigma_q)\longrightarrow
 \big((\cE^p)^{\otimes n}(\rho_{A_1\cdots A_n}),\sigma_q\big)
 \quad\text{for every }q.
 \label{eq:relative-majorization}
\end{equation}
\end{theorem}
\begin{proof}
For each $S$, choose a spectral decomposition into pure states,
\begin{equation}
 \rho_S=\sum_j r_{S,j}\phi_{S,j}.
 \label{eq:marginal-spectral}
\end{equation}
Because $\psi^{\otimes|S|}$ and $\phi_{S,j}$ are pure states on the same Hilbert space, there is a unitary $U_{S,j}$ satisfying
\begin{equation}
 U_{S,j}\psi^{\otimes|S|}U_{S,j}^\dagger=\phi_{S,j}.
 \label{eq:sector-unitary}
\end{equation}
Define the random-unitary channel
\begin{equation}
 \cR_S^\rho(X)\coloneqq\sum_j r_{S,j}U_{S,j}XU_{S,j}^\dagger.
 \label{eq:sector-random-unitary}
\end{equation}
For the empty sector, use the identity channel on scalars. Each channel in \eqref{eq:sector-random-unitary} is trace preserving and unital, and it satisfies
\begin{equation}
 \cR_S^\rho(\psi^{\otimes|S|})=\rho_S,
 \qquad\cR_S^\rho(I_S/D_S)=I_S/D_S.
 \label{eq:sector-actions}
\end{equation}

Let $P_S$ be the projector onto erasure sector $S$. Define $\cM^\rho$ by first measuring the sector and then applying $\cR_S^\rho$ within it:
\begin{equation}
 \cM^\rho(X)\coloneqq\bigoplus_S\cR_S^\rho(P_SXP_S),
 \label{eq:full-sector-channel}
\end{equation}
where each compressed block is identified with its surviving-system space. This is a channel on the full output space. Equations~\eqref{eq:direct-sum} and \eqref{eq:sector-actions} give both identities in \eqref{eq:simultaneous-conversion}. The construction of $\cM^\rho$ does not depend on $q$.
\end{proof}

\begin{figure}[t]
 \centering
 \input{figures/sector_conversion.tex}
 \caption{The relative-majorization channel acts separately in each erasure sector. A mixture of unitaries takes the pure reference direction to the desired marginal, while fixing the normalized identity. Consequently, every choice of auxiliary pattern weights $q_S$ is preserved by the same channel. The operations act on Bob's full surviving-system space and may entangle those systems.}
 \label{fig:sector-conversion}
\end{figure}

\subsection{The hypothesis-testing converse}
The invariant family becomes useful when a decoder is viewed as a binary hypothesis test. We recall the testing quantity, explain its basic properties, and then state the resulting no-feedback converse as a proposition. The hypothesis-testing relative entropy originates in the work of Buscemi and Datta~\cite{BD10} and Wang and Renner~\cite{WR12}. We use the convention that its error parameter bounds rejection of the first hypothesis.

For states $\omega$ and $\sigma$ and $\eps\in[0,1)$, define
\begin{equation}
 \begin{aligned}
 \beta_{\eps}(\omega\|\sigma)
 &\coloneqq\min\!\left\{\Tr[T\sigma]:0\leq T\leq I,\ \Tr[T\omega]\geq1-\eps\right\},\\
 \Htest^\eps(\omega\|\sigma)
 &\coloneqq-\log_2\beta_{\eps}(\omega\|\sigma).
 \end{aligned}
 \label{eq:testing-def}
\end{equation}
The POVM $(T,I-T)$ accepts the first hypothesis on outcome $T$. Thus $\beta_{\eps}$ is the least probability of incorrectly accepting that hypothesis when the state is $\sigma$, subject to rejecting $\omega$ with probability at most $\eps$. We use the extended-value convention $-\log_2 0=+\infty$.

Two properties will be used below; see also~\cite{KLW26}. First, applying a quantum channel before a test cannot improve distinguishability:
\begin{equation}
\Htest^\eps(\omega\|\sigma) \geq  \Htest^\eps(\mathcal N(\omega)\|\mathcal N(\sigma))
 .
 \label{eq:testing-dpi}
\end{equation}
Indeed, if $T$ is a feasible output test, then $\mathcal N^\dagger(T)$ is a feasible input test with the same two acceptance probabilities. Complete positivity and unitality of $\mathcal N^\dagger$ ensure that this operator lies between zero and the identity. Hence the minimum type-II error on the input is no greater than that on the output, giving \eqref{eq:testing-dpi} after applying $-\log_2$. Second, adjoining the same independent state $\eta$ to both hypotheses does not change the divergence:
\begin{equation}
 \Htest^\eps(\eta\otimes\omega\|\eta\otimes\sigma)
 =\Htest^\eps(\omega\|\sigma).
 \label{eq:testing-independent-factor}
\end{equation}
One inequality follows by the channel that prepares $\eta$, and the reverse inequality follows by discarding its system. These are the data-processing and stability properties needed in the proposition.

\begin{proposition}[Reduction of the no-feedback converse]\label{prop:reduced-meta}
Suppose a code transmits $M$ equiprobable messages through $(\cE^p)^{\otimes n}$ without feedback and with average error at most $\eps$. For the pure product reference output $\omega_\psi$ in \eqref{eq:pure-reference} and the auxiliary family in \eqref{eq:aux-family},
\begin{equation}
 \log_2M\leq\inf_q\Htest^\eps(\omega_\psi\|\sigma_q).
 \label{eq:reduced-meta}
\end{equation}
The bound also applies to a code whose maximum error is at most $\eps$.
\end{proposition}
\begin{proof}
For the given code, define
\begin{equation}
 \Omega_{XB^n}\coloneqq\frac1M\sum_{m=1}^{M}\proj m_X\otimes
 (\cE^p)^{\otimes n}(\rho^m_{A_1\cdots A_n}),
 \qquad\pi_X\coloneqq\frac1M\sum_{m=1}^M\proj m_X.
 \label{eq:code-cq-state}
\end{equation}
The classical register $X$ records the true message. The decoding POVM $(\Lambda_m)_m$ defines the test
\begin{equation}
 T_{XB^n}\coloneqq\sum_m\proj m_X\otimes\Lambda_m.
 \label{eq:decoder-test}
\end{equation}
Each effect lies between zero and the identity, so $0\leq T_{XB^n}\leq I_{XB^n}$. Against the first hypothesis, its acceptance probability is exactly the average decoding success. Against the second hypothesis, the message register and the channel output are independent. For every auxiliary output state $\sigma$,
\begin{equation}
 \begin{aligned}
 \Tr[T_{XB^n}\Omega_{XB^n}]&\geq1-\eps,\\
 \Tr[T_{XB^n}(\pi_X\otimes\sigma)]
 &=\frac1M\Tr\!\left[\Big(\sum_m\Lambda_m\Big)\sigma\right]
 =\frac1M.
 \end{aligned}
 \label{eq:decoder-test-probabilities}
\end{equation}
The test in \eqref{eq:decoder-test} is therefore a feasible competitor in the minimization defining $\beta_\eps(\Omega_{XB^n}\|\pi_X\otimes\sigma)$, with objective value $1/M$. The minimum cannot be larger, so applying the decreasing function $-\log_2$ gives $\Htest^\eps(\Omega_{XB^n}\|\pi_X\otimes\sigma)\geq\log_2M$. Because this holds for every normalized auxiliary state, it also holds after minimizing:
\begin{equation}
 \log_2M\leq\inf_\sigma
 \Htest^\eps(\Omega_{XB^n}\|\pi_X\otimes\sigma).
 \label{eq:cq-meta-converse}
\end{equation}
This is the decoder-to-test argument underlying the one-shot converse of~\cite{WR12}.

For each codeword, Theorem~\ref{thm:relative} supplies a channel $\cM^{\rho^m}$. Apply the channel that reads the classical register $X$ and then applies the corresponding $\cM^{\rho^m}$ to the output system. It takes $\pi_X\otimes\omega_\psi$ to $\Omega_{XB^n}$, while fixing $\pi_X\otimes\sigma_q$ for every $q$. Equations~\eqref{eq:testing-dpi} and~\eqref{eq:testing-independent-factor} consequently give
\begin{equation}
 \begin{aligned}
 \Htest^\eps(\Omega_{XB^n}\|\pi_X\otimes\sigma_q)
 &\leq\Htest^\eps(\pi_X\otimes\omega_\psi\|\pi_X\otimes\sigma_q)\\
 &=\Htest^\eps(\omega_\psi\|\sigma_q).
 \end{aligned}
 \label{eq:controlled-dpi}
\end{equation}
For every $q$, the left side is at least $\log_2M$ by the same feasible-test argument. Taking the infimum of the right side over $q$ proves \eqref{eq:reduced-meta}. Finally, a maximum-error bound implies the same average-error bound, so the proposition applies under either criterion.
\end{proof}

\begin{remark}[Why preserving a common auxiliary state is necessary]\label{rem:common-aux}
The minimization over $\sigma$ in \eqref{eq:cq-meta-converse} does not justify applying arbitrary message-dependent channels to the second hypothesis. Such channels can produce a state correlated with $X$, rather than one of the required form $\pi_X\otimes\sigma$. The simultaneous invariance in \eqref{eq:simultaneous-conversion} is precisely what prevents this problem.
\end{remark}

\subsection{An explicit optimizer, without a minimax exchange}\label{subsec:testing-opt}
Proposition~\ref{prop:reduced-meta} removes the optimization over codewords, leaving only the auxiliary distribution $q$. We now solve that remaining optimization. The proof gives both a test that works for every auxiliary distribution and an explicit distribution attaining the converse value.

The remaining testing problem has only one relevant acceptance probability in each erasure sector. Both hypotheses can be pinched into sectors without affecting the optimum. Within a sector, the first hypothesis is pure and the second is a multiple of the identity. Any part of a test orthogonal to the pure direction adds type-II cost without increasing acceptance of the first hypothesis. Hence
\begin{equation}
 \beta_{\eps}(\omega_\psi\|\sigma_q)
 =\min\!\left\{\sum_Sq_S\frac{t_S}{D_S}:
 0\leq t_S\leq1,\ \sum_Sw_St_S\geq1-\eps\right\}.
 \label{eq:testing-LP}
\end{equation}
For completeness, if a sector test is $T_S$, then
$t_S\coloneqq\Tr[T_S\psi^{\otimes|S|}]$ satisfies
$\Tr[T_S]\geq t_S$; replacing it by $t_S\psi^{\otimes|S|}$ preserves acceptance and cannot increase the cost. This proves \eqref{eq:testing-LP} directly.

\begin{theorem}[Exact testing value and auxiliary state]\label{thm:testing-value}
For $R_\eps$ in \eqref{eq:R-definition},
\begin{equation}
 \inf_q\Htest^\eps(\omega_\psi\|\sigma_q)=\log_2R_\eps.
 \label{eq:testing-value}
\end{equation}
If $j_*$ attains the minimum in \eqref{eq:R-definition}, an optimal distribution is
\begin{equation}
 q_S^*\coloneqq\frac{w_SD_S\one\{|S|\leq j_*\}}{Z_{j_*}}.
 \label{eq:q-optimal}
\end{equation}
In particular,
\begin{equation}
 \beta_{\eps}(\omega_\psi\|\sigma_{q^*})
 =\frac{B_{j_*}-\eps}{Z_{j_*}}=\frac1{R_\eps}.
 \label{eq:beta-optimal}
\end{equation}
\end{theorem}
\begin{proof}
The function $e(r)$ in \eqref{eq:real-error} is continuous, nondecreasing, and tends to one as $r\to\infty$. Since $R_\eps$ is its largest feasible argument and $\eps<1$, continuity gives $e(R_\eps)=\eps$, also when $\eps=0$. Set
\begin{equation}
 t_S^*\coloneqq\min\{1,D_S/R_\eps\}.
 \label{eq:test-saddle}
\end{equation}
The identity $1-e(R_\eps)=1-\eps$ shows that $(t_S^*)_S$ is feasible in \eqref{eq:testing-LP}. For every distribution $q$,
\begin{equation}
 \beta_{\eps}(\omega_\psi\|\sigma_q)
 \leq\sum_Sq_S\frac{t_S^*}{D_S}
 \leq\frac1{R_\eps}.
 \label{eq:beta-uniform-upper}
\end{equation}

Now use $q^*$ from \eqref{eq:q-optimal}. Its normalization follows from the definition of $Z_{j_*}$, which is positive because $B_{j_*}>\eps$. Its type-II cost in \eqref{eq:testing-LP} is
\begin{equation}
 \frac1{Z_{j_*}}\sum_{S:\,|S|\leq j_*}w_St_S.
 \label{eq:truncated-cost}
\end{equation}
The sectors with $|S|>j_*$ can provide at most $1-B_{j_*}$ of the required acceptance. Therefore, every feasible test must satisfy
\begin{equation}
 \sum_{S:\,|S|\leq j_*}w_St_S\geq B_{j_*}-\eps.
 \label{eq:low-sector-acceptance}
\end{equation}
Equality is feasible: accept all higher sectors and use the constant acceptance $(B_{j_*}-\eps)/B_{j_*}$ in the lower sectors. This constant lies in $[0,1]$. Hence the minimum cost is $(B_{j_*}-\eps)/Z_{j_*}=1/R_\eps$. Together with \eqref{eq:beta-uniform-upper}, this proves that the maximum over $q$ of the testing error is $1/R_\eps$. Applying $-\log_2$ gives \eqref{eq:testing-value}.
\end{proof}

The total auxiliary weight on sectors with $k$ survivors is $b_kd^k/Z_{j_*}$ for $k\leq j_*$ and zero otherwise. Thus the optimizer is a truncated, dimension-weighted distribution on the erasure patterns. It is the direct quantum sector analogue of an erasure-channel auxiliary distribution. The proof gives a feasible test for every $q$ and a matching auxiliary state, so no interchange of optimizations is needed.

Combining Proposition~\ref{prop:reduced-meta} and Theorem~\ref{thm:testing-value} gives
\begin{equation}
 C_{n,\eps}^{\nf}\leq\log_2\lfloor R_\eps\rfloor
 =C_{n,\eps}^{\fb}.
 \label{eq:nf-capacity-bound}
\end{equation}
Equivalently, $\succp^{*,\nf}(n,M)\leq s_{n,d,p}(M)$. The same numerical converse follows more directly by restricting a decoder to each erasure sector and applying \eqref{eq:elementary-dimension}. As noted in the introduction, this no-feedback bound is already present in~\cite[Eq.~(34)]{Rosati24}.

\subsection{The fully optimized classical--quantum converse}
The preceding theorem optimizes the sector-uniform auxiliary family for one reference output. We next show that it also evaluates the usual converse optimized over all finite input ensembles and all auxiliary output states. This distinguishes an exact optimization identity from an exact operational capacity formula. For a channel $\mathcal N$, define
\begin{equation}
 \chi_H^\eps(\mathcal N)\coloneqq
 \sup_{(p_x,\rho^x)_x}\ \inf_\sigma
 \Htest^\eps\!\left(
 \sum_xp_x\proj x_X\otimes\mathcal N(\rho^x)
 \middle\|\sum_xp_x\proj x_X\otimes\sigma\right),
 \label{eq:fully-optimized-meta}
\end{equation}
where the supremum is over finite input ensembles and the infimum is over output states. This quantity is a converse expression; it is not defined as an operational capacity.

\begin{proposition}[Exact value of the optimized converse]\label{prop:full-meta}
For the erasure channel,
\begin{equation}
 \chi_H^\eps((\cE^p)^{\otimes n})=\log_2R_\eps.
 \label{eq:full-meta-value}
\end{equation}
The uniform ensemble of the $d^n$ computational-basis product states attains the supremum in \eqref{eq:fully-optimized-meta}.
\end{proposition}
\begin{proof}
The upper bound follows from Theorem~\ref{thm:relative} and the controlled-channel argument for an arbitrary ensemble, followed by Theorem~\ref{thm:testing-value}. For the uniform computational-basis ensemble, sector pinching, phase averaging, and simultaneous shifts of the classical label and output systems reduce every auxiliary state to the family $\sigma_q$ without increasing the testing divergence. Controlled inverse shifts then convert the first hypothesis to a product of the uniform label state and $\omega_\psi$, with $\psi=\proj0$, while fixing the second hypothesis. Theorem~\ref{thm:testing-value} gives the matching lower bound. Appendix~\ref{app:symmetry} provides the explicit channels and both optimization inequalities.
\end{proof}

Thus even the fully optimized converse in \eqref{eq:fully-optimized-meta} can have a strict operational gap. The next section proves this using only four channel uses.

\section{A strict gap without feedback}\label{sec:strict}
Section~\ref{sec:exact} determines an exact operational optimum with classical feedback, whereas Section~\ref{sec:majorization} evaluates a converse for communication without feedback. The numerical values of these two expressions agree. This agreement does not yet show that the no-feedback converse is achievable: an exactly evaluated bound need not be attained by a communication code. We now establish a concrete example where it is not.

The distinction arises from the requirements imposed by different erasure patterns. If a pattern were specified before encoding, one could choose states tailored to its surviving systems and attain the corresponding dimension bound. Without feedback, however, Alice must prepare one codeword for each message before learning the pattern. Its reduced states on every possible survivor set must all come from that same codeword. Bob may still read the erasure flags and choose a different decoding measurement for each pattern; the restriction is the compatibility of the reduced code states, not a restriction on his measurement.

For four messages and four qubit transmissions, we show that equality in the average dimension bound would force incompatible distinguishability conditions on the two-qubit marginals. This proves a strict gap in optimal success probability and, at a particular error tolerance, an exact separation between the largest message sizes with and without feedback. Section~\ref{sec:tetra} then addresses a different question: whether quantum encoding can nevertheless outperform classical encoding when neither has feedback.

\subsection{Four messages and four qubit transmissions}
Every two-survivor pattern in this example leaves a four-dimensional output space. Attaining the dimension bound would therefore require perfect discrimination of the four messages on every pair of surviving qubits. We first translate that requirement into purity and orthogonality conditions, then show that no collection of four codewords can satisfy all of them.

\begin{theorem}[Strict unassisted gap]\label{thm:strict-gap}
Let $\succp^{*,\nf}(n,M)$ be the optimal average success probability without feedback, defined in Section~\ref{subsec:codes} by maximizing \eqref{eq:code-success}, and let $s_{n,d,p}(M)$ be the dimension-counting success function in \eqref{eq:success-function}. For $n=4$, $d=2$, $M=4$, and $0<p<1$,
\begin{equation}
 \succp^{*,\nf}(4,4)<s_{4,2,p}(4).
 \label{eq:strict-gap}
\end{equation}
Thus no four-message code using four qubit erasure-channel transmissions without feedback attains the success probability achievable with classical feedback, which equals the right-hand side by Theorem~\ref{thm:feedback}. The inequality remains strict even when the codewords are mixed and entangled across all four inputs and the receiver makes an arbitrary collective measurement.
\end{theorem}
\begin{proof}
Assume a code attains the bound. Reading the orthogonal erasure sectors first leaves its success probability unchanged. Every pattern has positive probability because $0<p<1$, so equality in the weighted sum of patternwise dimension bounds forces equality for each pattern.

On each two-survivor pattern, the dimension is four and the bound is one. Attaining unit average success forces unit success for every message, so the four reduced states must be perfectly distinguishable. Perfectly distinguishable states have orthogonal supports. To see this directly, an effect that identifies one state with probability one acts as the identity on that state's support; its zero probability on another state forces it to vanish on the other's support. Four nonzero orthogonal supports in a four-dimensional space must all have dimension one. Consequently, every two-qubit marginal of every codeword is pure.

This purity condition forces each codeword to be a pure product of four single-qubit states. Indeed, a state with a pure marginal factors across that marginal and its complement. This follows, for example, by observing that its support on that subsystem is one dimensional: positivity rules out support outside that subspace and all cross terms with it. Applied to the marginal on sites $12$, this gives $\rho_{1234}=\phi_{12}\otimes\rho_{34}$. Purity of the $34$ marginal makes the second factor pure. The marginal on $13$ is then $\rho_1\otimes\rho_3$, whose purity forces both single-site states to be pure. The pure states on $12$ and $34$ therefore factor as well. Write the four codewords as
\begin{equation}
 \rho^m=\bigotimes_{i=1}^4\proj{u_{m,i}},
 \qquad m\in\{1,2,3,4\}.
 \label{eq:forced-product-code}
\end{equation}

Fix distinct messages $m,m'$. Orthogonality of their reductions on every pair $i,j$ gives
\begin{equation}
 \langle u_{m,i}|u_{m',i}\rangle
 \langle u_{m,j}|u_{m',j}\rangle=0
 \quad\text{for every }i\neq j.
 \label{eq:pairwise-orthogonality}
\end{equation}
There can therefore be at most one site at which the two local vectors are nonorthogonal. Each pair of messages must be locally orthogonal at at least three sites. There are six unordered message pairs, so this requires at least $6\cdot3=18$ incidences of a message pair being orthogonal at a site.

At a fixed qubit site, however, among four normalized vectors there can be at most four orthogonal pairs. If there were at least five such pairs, the graph whose vertices are the four messages and whose edges are the orthogonal pairs would contain a triangle. That would give three mutually orthogonal vectors in a two-dimensional space, which is impossible. Across four sites, there are therefore at most $4\cdot4=16$ orthogonality incidences. This contradicts the required eighteen.

We have proved nonattainability of the dimension-counting value. To obtain a strict inequality for the optimum, note that the set of four input density operators and the set of four-outcome POVMs on the finite output space are compact. The average success probability is continuous in these variables, so its maximum is attained. The preceding contradiction therefore proves \eqref{eq:strict-gap}, not merely the absence of an equality-achieving code in a noncompact class.
\end{proof}

Allowing shared randomness without feedback does not improve the optimum average success: a randomized code is an average of fixed codes, and one of those fixed codes has success at least as large as the average. Thus the strict average-success gap also survives such assistance. No assertion that the optimal unassisted average and maximum errors are equal is needed here.

\subsection{An exact separation of message sizes}
The preceding strict inequality yields an explicit difference in finite-blocklength capacities, even though the unrestricted optimal four-message success probability remains undetermined.

\begin{corollary}[An exact capacity separation]\label{cor:capacity-gap}
For $n=4$, $d=2$, $p=1/2$, and $\eps=11/64$,
\begin{equation}
 C_{4,11/64}^{\nf}=\log_2 3,
 \qquad C_{4,11/64}^{\fb}=2.
 \label{eq:exact-capacity-gap}
\end{equation}
The first equality holds for both average and maximum error; the second holds under the feedback conventions of Theorem~\ref{thm:feedback}.
\end{corollary}
\begin{proof}
The feedback formula gives $s_{4,2,1/2}(4)=53/64$, so four messages are feasible at error $11/64$. For five messages it gives $s_{4,2,1/2}(5)=29/40<53/64$. Since the function is nonincreasing in the message size, five or more messages are impossible even with feedback.

Without feedback, Theorem~\ref{thm:strict-gap} excludes four messages. Three messages are achievable by the computational-basis codewords
\begin{equation}
 0000,\qquad0011,\qquad1101.
 \label{eq:three-message-code}
\end{equation}
Bob guesses uniformly among the codewords consistent with the unerased coordinates. The pairwise distances are $2,3,3$, and all three codewords become simultaneously indistinguishable only when all four coordinates are erased. The individual errors are therefore
\begin{equation}
 \eps_1=\eps_2=\frac12(p^2+p^3)-\frac13p^4,
 \qquad\eps_3=p^3-\frac13p^4.
 \label{eq:three-message-errors}
\end{equation}
To check the correction term, two pairwise confusions would contribute $1/2+1/2=1$ on the all-erased event, whereas the actual three-way error is $2/3$; one subtracts $p^4/3$. At $p=1/2$, the errors are $1/6,1/6,5/48$, respectively. Their maximum is $1/6<11/64$, proving achievability under either error criterion.
\end{proof}

\begin{remark}[Cases in which the no-feedback bound is attainable]\label{rem:attainable-cases}
The strict gap is not universal. If $M\leq d$, repeating one of $M$ orthogonal symbols attains the bound. If $M\geq d^n$, a classical code that contains every length-$n$ word at least once, with repetitions for the remaining messages, attains the average-success bound: each $k$-coordinate projection has exactly $d^k$ possible outputs. More generally, a classical code with $M=d^r$ messages attains the bound if every projection onto $r$ coordinates is bijective. Projections onto fewer coordinates are then surjective, and projections onto more coordinates are injective. Such maximum-distance-separable constructions are relevant when their parameters permit them~\cite{PPV10,Rosati24}; they do not exist for arbitrary alphabet size and blocklength.
\end{remark}

\section{A product quantum code that outperforms every classical code}\label{sec:tetra}
The failure of universal attainability does not mean that the quantum no-feedback problem reduces to its classical counterpart. We now exhibit a quantum advantage using four product-state codewords. The gain comes from nonorthogonal signal states, not from entanglement among the channel inputs.

\subsection{Tetrahedral states and the four-carrier encoding}
Our code combines three orthogonal binary signals with one nonorthogonal qubit signal. The latter is chosen from a tetrahedron on the Bloch sphere, so that each pair of its signal states has the same distinguishability. We first define these states, then explain the decoder for each kind of erasure pattern. For this four-message example, it is convenient to write $m=1+2x+y$ with $x,y\in\{0,1\}$; the pair $(x,y)$ is a binary representation of $m-1$, not a change to the message alphabet.

Let $\sigma_1,\sigma_2,\sigma_3$ be the Pauli matrices. For $x,y\in\{0,1\}$, define
\begin{equation}
 \boldsymbol r_{xy}\coloneqq\frac1{\sqrt3}
 \big((-1)^x,(-1)^y,(-1)^{x+y}\big),
 \qquad
 \tau_{xy}\coloneqq\frac12\left(I+\sum_{j=1}^3 r_{xy,j}\sigma_j\right).
 \label{eq:tetra-states}
\end{equation}
The four Bloch vectors are vertices of a regular tetrahedron. These are the qubit example of symmetric informationally complete signal states~\cite{RBSC04}. Direct calculation gives
\begin{equation}
 \tau_{xy}^2=\tau_{xy},\qquad
 \sum_{x,y}\tau_{xy}=2I,\qquad
 \Tr[\tau_{xy}\tau_{x'y'}]=\frac13
 \quad\text{if }(x,y)\neq(x',y').
 \label{eq:tetra-identities}
\end{equation}
The last expression is the squared overlap of the corresponding pure-state vectors.

Encode the message $(x,y)$ as
\begin{equation}
 \rho^{xy}_{A_1A_2A_3A_4}
 \coloneqq\proj x_{A_1}\otimes\proj y_{A_2}
 \otimes\proj{x\oplus y}_{A_3}\otimes\tau_{xy,A_4},
 \label{eq:tetra-code}
\end{equation}
where $\oplus$ denotes addition modulo two. Two of the first three classical bits determine the message. If only one of those bits and the tetrahedral carrier survive, the classical bit reduces the decoding problem to distinguishing two tetrahedral states. Figure~\ref{fig:tetra} illustrates these two situations.

\begin{figure}[t]
 \centering
 \input{figures/tetrahedral_code.tex}
 \caption{The product-state encoding in \eqref{eq:tetra-code}. The first three carriers contain the classical bits $x$, $y$, and $x\oplus y$; the fourth contains a tetrahedral state. Any two of the first three carriers identify the message. A classical carrier together with the fourth carrier leaves a binary pure-state discrimination problem. The Bloch-sphere picture is schematic; the exact vectors are given in \eqref{eq:tetra-states}. No channel-input entanglement is used.}
 \label{fig:tetra}
\end{figure}
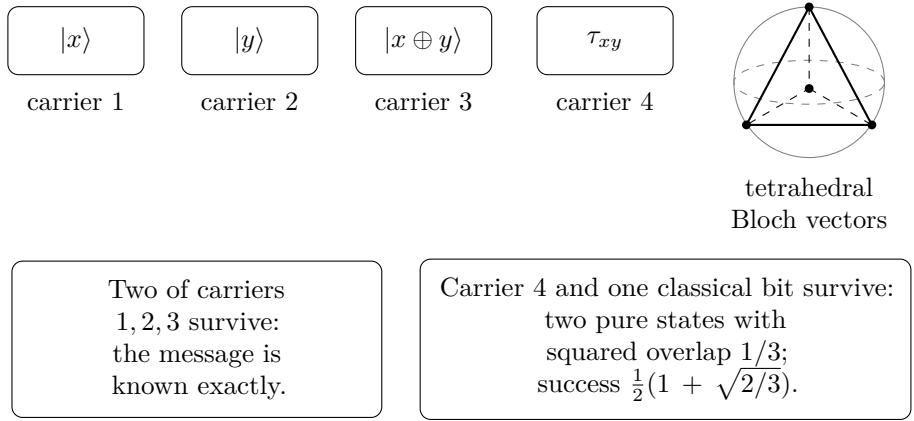

For two equiprobable pure states $\phi$ and $\varphi$, optimizing the first POVM effect $0\leq T\leq I$ gives
\begin{equation}
 \max_T\frac12\big(\Tr[T\phi]+\Tr[(I-T)\varphi]\big)
 =\frac12\big(1+\Tr[(\phi-\varphi)_+]\big)
 =\frac12\left(1+\sqrt{1-\Tr[\phi\varphi]}\right).
 \label{eq:binary-pure-success}
\end{equation}
The last equality follows by diagonalizing the traceless rank-at-most-two operator $\phi-\varphi$: its nonzero eigenvalues are $\pm\sqrt{1-\Tr[\phi\varphi]}$. Applied to a tetrahedral pair, this gives
\begin{equation}
 h\coloneqq\frac12\left(1+\sqrt{\frac23}\right).
 \label{eq:tetra-binary-success}
\end{equation}

\begin{proposition}[Achievable tetrahedral performance]\label{prop:tetra}
The code \eqref{eq:tetra-code} has a decoder achieving
\begin{equation}
 P_{\mathrm{tet}}(p)
 =s_{4,2,p}(4)-\frac{3-\sqrt6}{2}\,p^2(1-p)^2.
 \label{eq:tetra-performance}
\end{equation}
The individual-message success probabilities can all be made equal to this value.
\end{proposition}
\begin{proof}
Bob first reads the erasure pattern, which does not disturb the surviving systems. If no carrier survives, he guesses uniformly and succeeds with probability $1/4$. If exactly one of the first three carriers survives, it leaves two equiprobable messages, so uniform guessing succeeds with probability $1/2$. If only the fourth survives, the POVM $(\tau_{xy}/2)_{x,y}$ is complete by \eqref{eq:tetra-identities}, and its probability of the correct outcome is $1/2$ for every message.

If at least three carriers survive, at least two of the first three survive, so decoding is perfect. Among the six two-survivor patterns, three involve only the first three carriers and are also decoded perfectly. Each of the other three patterns contains the tetrahedral carrier and one classical carrier. Bob first reads the classical bit and then makes the optimal binary measurement for the two remaining tetrahedral states. Its average success is $h$ in \eqref{eq:tetra-binary-success}.

Every two-survivor pattern has weight $p^2(1-p)^2$. Thus the only loss from the dimension-counting bound is $3(1-h)p^2(1-p)^2$. Since $3(1-h)=(3-\sqrt6)/2$, this gives \eqref{eq:tetra-performance}. For two equiprobable pure states, the positive-eigenspace binary measurement has equal success probability on the two states; this also follows from the two eigenvalues used in \eqref{eq:binary-pure-success}. All other decoding rules just described are message symmetric. Therefore every individual message has the stated success probability.
\end{proof}

\subsection{The exact classical optimum for the same parameters}
We compare this code to the optimal code for four uses of a classical binary erasure channel and four equiprobable messages, without feedback. Randomized classical encoders do not improve the optimum average success. For a fixed decoder the success is linear in each message's input distribution, so one can select a deterministic input word for every message without decreasing success; optimizing the decoder afterward cannot reduce it.

\begin{theorem}[Optimal four-message classical binary erasure code]\label{thm:classical-four}
For four messages and four uses of the binary erasure channel,
\begin{equation}
 P_\cl^*(p)=s_{4,2,p}(4)-\frac12p^2(1-p)^2.
 \label{eq:classical-four-optimum}
\end{equation}
The formula holds for optimal average success and for the best worst-message success. It is attained by the code
\begin{equation}
 (x,y)\longmapsto(x,y,x\oplus y,x).
 \label{eq:classical-four-code}
\end{equation}
\end{theorem}
\begin{proof}
Fix four binary words, allowing repetitions. For a survivor set $S$, let $N_S$ be the number of distinct projected words on $S$. Every distinct output identifies a nonempty class of compatible messages. The optimal average success conditioned on $S$ is $N_S/4$, by the same class-counting argument used in \eqref{eq:digit-success}.

For zero, one, three, and four survivors, the sums of $N_S$ over all subsets of the specified size are bounded by $1$, $8$, $16$, and $4$, respectively. These are the direct dimension bounds. For two survivors, we claim
\begin{equation}
 \sum_{S:\,|S|=2}N_S\leq22.
 \label{eq:classical-pair-count}
\end{equation}
Consider how each coordinate partitions the four message labels by its bit value. If one coordinate is unbalanced, its partition has sizes $0,4$ or $1,3$, up to exchanging the bit labels. Each of the three two-coordinate projections containing that coordinate then has at most three distinct outputs: in the $1,3$ case, one class contributes at most one output and the other at most two. The remaining three projections have at most four outputs each. Their total is at most $3\cdot3+3\cdot4=21$.

Otherwise every coordinate defines a balanced $2,2$ partition. Up to exchanging the bit labels, there are exactly three such partitions of four labels. Four coordinates therefore include two that define the same partition. Their joint projection has only two outputs; each of the other five projections has at most four. This proves \eqref{eq:classical-pair-count}.

The code \eqref{eq:classical-four-code} attains all these coefficient bounds simultaneously. The pair of repeated $x$ coordinates has two outputs, and each of the other five pairs has four. Each single coordinate has two outputs, every three-coordinate projection identifies the message, and the full code has four distinct words. Hence its projection-count sums are
\begin{equation}
 \left(\sum_{|S|=k}N_S\right)_{k=0}^4=(1,8,22,16,4).
 \label{eq:classical-coefficients}
\end{equation}
For two survivors, the dimension-counting value would have total $24$ rather than $22$. Dividing the deficit by four messages and multiplying by the pattern weight $p^2(1-p)^2$ gives precisely the loss in \eqref{eq:classical-four-optimum}.

Finally, the code is linear in $(x,y)$. For each erasure pattern, every nonempty compatible-message class is a coset of the same kernel and therefore has the same size. Uniform guessing gives the same success probability for every message. Thus the optimal average value is achieved also under the worst-message criterion.
\end{proof}

\subsection{The three-way separation}
We can now compare the optimal classical code, the explicit quantum product-state code, and the optimal feedback-assisted protocol within a single expression. The quantum no-feedback optimum is not evaluated, but the preceding upper and lower bounds place it strictly between the classical and feedback-assisted values.

The strict inequality $\sqrt6>2$ implies, for every $0<p<1$,
\begin{equation}
 P_\cl^*(p)<P_{\mathrm{tet}}(p)
 \leq\succp^{*,\nf}(4,4)<\succp^{*,\fb}(4,4).
 \label{eq:all-p-hierarchy}
\end{equation}
The final strict inequality is Theorem~\ref{thm:strict-gap}. Equation~\eqref{eq:intro-hierarchy} follows by setting $p=1/2$.

The tetrahedral carrier answers any one of three possible binary discrimination questions reasonably well. Which question is relevant is determined by the surviving classical coordinate. By contrast, the fourth classical bit in \eqref{eq:classical-four-code} must repeat one of only three balanced partitions, leading to the projection deficit in \eqref{eq:classical-pair-count}. This explains how nonorthogonal quantum signals improve performance without requiring entangled codewords.

Figure~\ref{fig:gap-plot} plots the gaps from the exact feedback value. It distinguishes the performance of the explicit tetrahedral code from the unknown optimum over all quantum codes without feedback.

\begin{figure}[t]
 \centering
 \input{figures/performance_gap.tex}
 \caption{Success-probability deficits from the exact feedback value for four qubit erasure-channel uses and four messages. The classical curve is optimal without feedback; the tetrahedral curve is the deficit of the explicit product-state code, not a claimed optimum over quantum codes. For $0<p<1$, the optimal quantum no-feedback deficit is strictly positive and no larger than the tetrahedral curve. The zero curve is achieved with feedback.}
 \label{fig:gap-plot}
\end{figure}
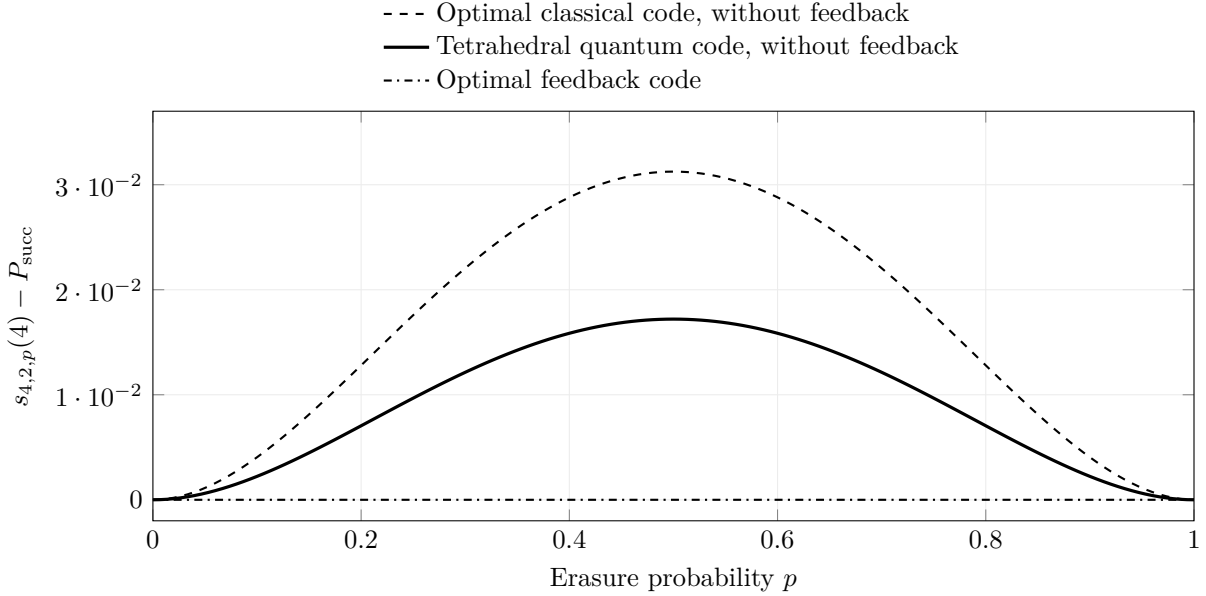

\section{Asymptotic consequences of the exact formula}\label{sec:asymptotic}
The exact feedback formula supplies converses for both communication models. Classical codes without feedback provide matching asymptotic achievability in several regimes, even though they need not attain the finite-blocklength formula. We first derive a bounded-remainder normal approximation with feedback, and then show that it remains valid without feedback. Next we determine the common strong-converse exponent, which concerns rates above capacity and the decay of success probability. Finally, we study error exponents below capacity: the feedback exponent is exact at every positive subcapacity rate, while the no-feedback exponent is determined above a critical rate and bounded below it.

The asymptotic capacity and the feedback-assisted strong-converse property were already established in~\cite{DQSW19,DKQSW23}; the unassisted quantum second-order expansion was established in~\cite[Theorem~13]{WFT19}. The classical constructions used below are specializations of established erasure-channel random coding~\cite{PPV10} and error-exponent methods~\cite{Gallager65}. We include proofs to specify the stronger remainder estimates, the applicable error criteria, and the rate ranges where the asymptotic comparisons are exact.

\subsection{A normal approximation with a bounded remainder}
The exact error in \eqref{eq:real-error} is a smoothed lower tail of the number of survivors. We first bound the difference from the binomial cumulative distribution function, then approximate that distribution by a Gaussian, and finally invert the resulting estimate. We include each step because an $O(n^{-1/2})$ error in a distribution function becomes an $O(1)$ error in the message-size threshold.

Let $\Phi$ and $\varphi$ be the cumulative distribution function and density of a standard normal random variable, respectively. Throughout this subsection, $p$ and $\eps$ are fixed in $(0,1)$, and $d\geq2$ is fixed.

\begin{corollary}[Normal approximation with a uniform bounded remainder]\label{cor:normal}
For every positive integer $n$,
\begin{equation}
 C_{n,\eps}^{\fb}
 =n(1-p)\log_2d
 +\sqrt{np(1-p)}\,\log_2d\,\Phi^{-1}(\eps)+O(1).
 \label{eq:normal-approximation}
\end{equation}
Precisely, there is a finite constant $A_{p,d,\eps}$, independent of $n$, such that
\begin{equation}
 \left|C_{n,\eps}^{\fb}-n(1-p)\log_2d
 -\sqrt{np(1-p)}\,\log_2d\,\Phi^{-1}(\eps)\right|
 \leq A_{p,d,\eps}
 \quad\text{for all }n\in\N.
 \label{eq:normal-uniform-remainder}
\end{equation}
\end{corollary}
\begin{proof}
Set $q\coloneqq1-p$, and write $K=\sum_{i=1}^n X_i$, where the $X_i$ are independent Bernoulli random variables with success probability $q$. Thus $\E[K]=nq$ and $\operatorname{Var}(K)=npq$. For a real threshold $r\geq0$, the corresponding real message size is $d^r$. Only terms with $k<r$ contribute to $e(d^r)$, and hence
\begin{equation}
 \begin{aligned}
 0\leq\Pr(K<r)-e(d^r)
 &=\sum_{\substack{0\leq k\leq n\\k<r}}\Pr(K=k)d^{k-r}\\
 &\leq\frac{d}{d-1}\max_{0\leq k\leq n}\Pr(K=k).
 \end{aligned}
 \label{eq:cdf-smoothing-error}
\end{equation}
To justify the last inequality, set $j(r)\coloneqq\lceil r\rceil-1$, the largest integer strictly smaller than $r$. If the sum is nonempty, replacing each probability by its largest value and extending the geometric sum to all integers $k\leq j(r)$ can only increase it. The resulting geometric factor is
\begin{equation}
 \sum_{k=-\infty}^{j(r)}d^{k-r}
 =\frac{d^{j(r)-r}}{1-d^{-1}}
 \leq\frac{d}{d-1}.
 \label{eq:geometric-smoothing-factor}
\end{equation}
An empty sum is zero. The strict threshold $k<r$ is important when $r$ is an integer: the term $k=r$ contributes zero to the error and is excluded from both sides of the difference in \eqref{eq:cdf-smoothing-error}.

We next state the probability estimate used to control both the cumulative distribution and its largest atom. The Berry--Esseen theorem asserts that there is a universal finite constant $C_{\mathrm{BE}}$ such that, for independent identically distributed real random variables $Y_i$ with mean $\mu$, variance $\sigma^2>0$, and third absolute centered moment $\tau<\infty$,
\begin{equation}
 \sup_{z\in\mathbb R}
 \left|\Pr\!\left(\frac{\sum_{i=1}^nY_i-n\mu}{\sigma\sqrt n}\leq z\right)
 -\Phi(z)\right|
 \leq\frac{C_{\mathrm{BE}}\tau}{\sigma^3\sqrt n}.
 \label{eq:berry-esseen-general}
\end{equation}
We use only the existence of such a constant, not its optimal value; see~\cite{Berry41,Esseen42} and~\cite[Theorem~44]{PPV10}. For a Bernoulli variable $X_i$, the two possible centered values are $p$ and $-q$, with probabilities $q$ and $p$. Consequently,
\begin{equation}
 \E[|X_i-q|^3]=qp^3+pq^3=pq(p^2+q^2).
 \label{eq:bernoulli-third-moment}
\end{equation}
Define $a_p\coloneqq C_{\mathrm{BE}}(p^2+q^2)/\sqrt{pq}$. Substitution in \eqref{eq:berry-esseen-general} gives
\begin{equation}
 \left|\Pr(K\leq x)
 -\Phi\!\left(\frac{x-nq}{\sqrt{npq}}\right)\right|
 \leq\frac{a_p}{\sqrt n}
 \quad\text{for every real }x.
 \label{eq:berry-esseen-binomial}
\end{equation}

This uniform estimate also bounds the largest point probability, so a separate use of Stirling's formula is unnecessary. Set $z_k\coloneqq(k-nq)/\sqrt{npq}$. Subtracting the cumulative probabilities at adjacent integers and using $\sup_z\varphi(z)=1/\sqrt{2\pi}$ yields
\begin{equation}
 \begin{aligned}
 \Pr(K=k)
 &\leq\Phi(z_k)-\Phi(z_{k-1})+\frac{2a_p}{\sqrt n}\\
 &\leq\frac{u_p}{\sqrt n},
 \qquad
 u_p\coloneqq2a_p+\frac1{\sqrt{2\pi pq}}.
 \end{aligned}
 \label{eq:binomial-largest-atom}
\end{equation}
The same argument covers $k=0$ and $k=n$, since the cumulative distribution is defined outside the support of $K$ as well.

For the strict inequality in \eqref{eq:cdf-smoothing-error}, use $\Pr(K<r)=\Pr(K\leq j(r))$. The distance from $j(r)$ to $r$ is at most one. Applying \eqref{eq:berry-esseen-binomial} at $j(r)$ and bounding the change in $\Phi$ over this one-step displacement gives
\begin{equation}
 \left|\Pr(K<r)
 -\Phi\!\left(\frac{r-nq}{\sqrt{npq}}\right)\right|
 \leq\frac{v_p}{\sqrt n},
 \qquad
 v_p\coloneqq a_p+\frac1{\sqrt{2\pi pq}}.
 \label{eq:binomial-normal}
\end{equation}
Thus the discrete threshold causes an error of the same order as the Berry--Esseen estimate, uniformly in the real threshold. Combining \eqref{eq:cdf-smoothing-error}, \eqref{eq:binomial-largest-atom}, and \eqref{eq:binomial-normal}, we obtain
\begin{equation}
 \left|e(d^r)
 -\Phi\!\left(\frac{r-nq}{\sqrt{npq}}\right)\right|
 \leq\frac{c_{p,d}}{\sqrt n},
 \qquad
 c_{p,d}\coloneqq v_p+\frac{d}{d-1}u_p.
 \label{eq:smoothed-error-normal}
\end{equation}

It remains to invert this estimate. Set $r_n\coloneqq\log_dR_\eps$ and $z_\eps\coloneqq\Phi^{-1}(\eps)$. By the definition of $R_\eps$ and continuity of $e$, we have $e(R_\eps)=\eps$. Equation~\eqref{eq:smoothed-error-normal} therefore implies
\begin{equation}
 \left|\Phi\!\left(\frac{r_n-nq}{\sqrt{npq}}\right)-\eps\right|
 \leq\frac{c_{p,d}}{\sqrt n}.
 \label{eq:quantile-step}
\end{equation}
The two numbers $\Phi(z_\eps)-\Phi(z_\eps-1)$ and $\Phi(z_\eps+1)-\Phi(z_\eps)$ are strictly positive. Choose $n_0$ so that the right side of \eqref{eq:quantile-step} is smaller than both whenever $n\geq n_0$. Monotonicity of $\Phi$ then places $(r_n-nq)/\sqrt{npq}$ in $[z_\eps-1,z_\eps+1]$. On this compact interval the normal density has the positive minimum
\begin{equation}
 c_\eps\coloneqq\min\!\left\{\varphi(z):|z-z_\eps|\leq1\right\}>0.
 \label{eq:quantile-density-lower}
\end{equation}
The mean value theorem applied to \eqref{eq:quantile-step} now gives the explicit threshold bound
\begin{equation}
 \left|r_n-nq-\sqrt{npq}\,z_\eps\right|
 \leq\frac{c_{p,d}\sqrt{pq}}{c_\eps}
 \quad(n\geq n_0).
 \label{eq:real-message-normal}
\end{equation}
Multiplication by $\log_2d$ bounds the remainder in $\log_2R_\eps$. Finally, $R_\eps\geq1$ and $\lfloor R_\eps\rfloor\geq R_\eps/2$, so replacing $R_\eps$ by its integer part changes the binary logarithm by at most one bit. Corollary~\ref{cor:capacity} therefore gives a constant remainder bound for $n\geq n_0$. Enlarging that constant to cover the finitely many integers $1\leq n<n_0$ proves \eqref{eq:normal-uniform-remainder} for every positive blocklength.
\end{proof}

The second-order coefficient in \eqref{eq:normal-approximation} is $\sqrt{np(1-p)}\,\log_2d$ multiplied by $\Phi^{-1}(\eps)$. The uniform bound is an all-blocklength statement, but it does not assert that the remainder tends to zero or converges to a third-order constant. The discreteness of the survivor count and the integer message size can leave bounded oscillations.

Figure~\ref{fig:normal-rate} compares the exact rate $C_{n,0.1}^{\fb}/n$ with the first two terms of \eqref{eq:normal-approximation}, divided by $n$, for $d=2$, $p=1/2$, and $\eps=0.1$. The approximation captures the leading $n^{-1/2}$ backoff from capacity. The $O(1)$ remainder in total bits becomes an $O(n^{-1})$ remainder in this rate plot. The two-term curve is an approximation, not a lower or upper bound asserted for every $n$.

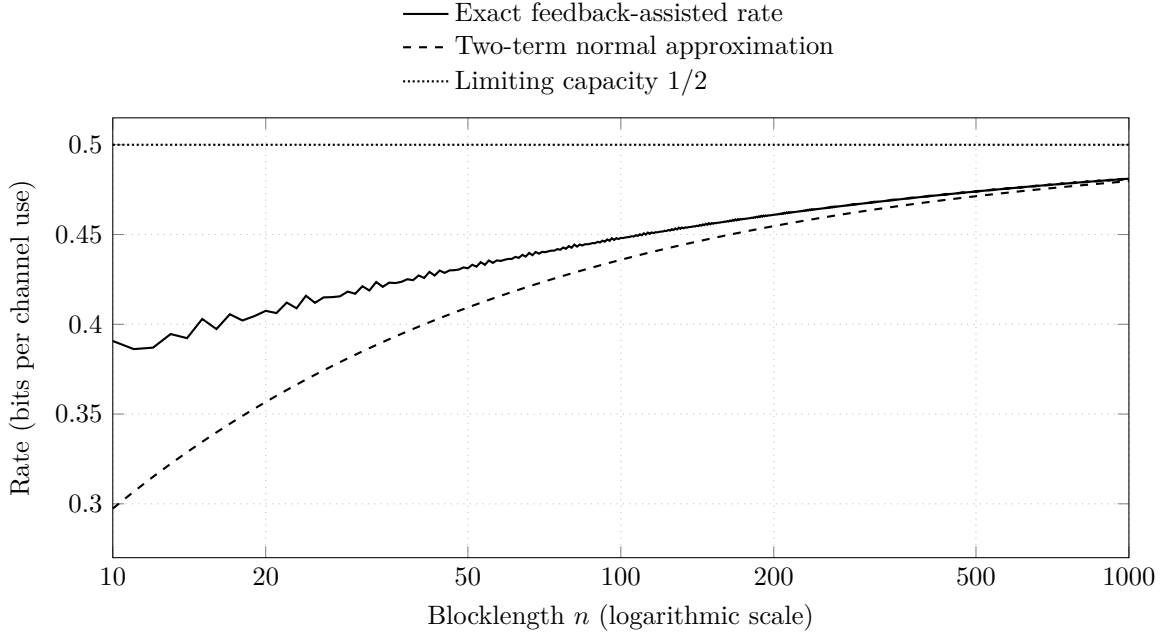
\begin{figure}[t]
 \centering
 \input{figures/normal_rate.tex}
 \caption{Normal approximation in Corollary~\ref{cor:normal}, for $d=2$, $p=1/2$, and $\eps=0.1$. The solid curve is the exact feedback-assisted rate from Corollary~\ref{cor:capacity}. The dashed curve is $\frac12+\Phi^{-1}(0.1)/(2\sqrt n)$, with no fitted constant. The horizontal dotted line is the limiting capacity $1/2$ bit per use. The blocklength axis is logarithmic. These finite-length exact values are feedback-assisted; the plot does not claim to evaluate the no-feedback optimum.}
 \label{fig:normal-rate}
\end{figure}

\FloatBarrier
\subsection{Classical codes without feedback}\label{subsec:nf-asymptotics}
The feedback formula is also a converse for codes without feedback, but a matching asymptotic lower bound still requires a no-feedback construction. For this purpose, ordinary random classical codewords suffice. We first count the different projected words that a receiver can see, and then control the collisions between such projections. These are the classical erasure-channel random-coding arguments underlying~\cite[Theorems~36--37]{PPV10}; we give their $d$-ary forms to make their application to quantum channels explicit.

Let $\succp^{*,\cl,\nf}(n,M)$ denote the optimal average success probability when all channel inputs are computational-basis states. Equivalently, this is the optimum for a classical $d$-ary erasure channel without feedback. Allowing arbitrary quantum codewords gives $\succp^{*,\nf}(n,M)$ as before.

\begin{proposition}[Classical random-coding bounds]\label{prop:random-classical}
For every positive integer $M$, define
\begin{equation}
 a_{n,d,p}(M)\coloneqq
 \sum_{k=0}^n b_k\frac{d^k}{M}\left[1-(1-d^{-k})^M\right].
 \label{eq:occupancy-success}
\end{equation}
Then
\begin{equation}
 (1-e^{-1})s_{n,d,p}(M)
 \leq a_{n,d,p}(M)
 \leq\succp^{*,\cl,\nf}(n,M)
 \leq\succp^{*,\nf}(n,M)
 \leq s_{n,d,p}(M).
 \label{eq:constant-success-bracket}
\end{equation}
There also exists a deterministic classical code with average error at most
\begin{equation}
 U_n(M)\coloneqq
 \sum_{k=0}^n b_k\min\!\left\{1,\frac{M-1}{2d^k}\right\}.
 \label{eq:random-error-bound}
\end{equation}
No shared randomness is needed to implement either code.
\end{proposition}
\begin{proof}
Draw $M$ words of length $n$ independently and uniformly from $\{0,\ldots,d-1\}^n$, allowing repetitions. Fix a survivor set of size $k$ and set $D=d^k$. Each of the $D$ possible projected words is absent with probability $(1-1/D)^M$. The expected number of distinct projected words is therefore $D[1-(1-1/D)^M]$. For a fixed codebook, the optimal average success on this pattern is the number of distinct projected words divided by $M$: each nonempty compatible-message class contributes $1/M$. Averaging over patterns gives \eqref{eq:occupancy-success}. At least one deterministic codebook attains this ensemble-average value.

To obtain the first inequality in \eqref{eq:constant-success-bracket}, set $x=M/D>0$ and use
\begin{equation}
 1-(1-1/D)^M\geq1-e^{-x}
 \geq(1-e^{-1})\min\{1,x\}.
 \label{eq:occupancy-elementary}
\end{equation}
For $0\leq x\leq1$, the last inequality is the chord bound for the concave function $1-e^{-x}$ between zero and one. For $x\geq1$, it follows from monotonicity. Multiplication by $D/M$ and averaging gives the claim. The remaining comparisons follow from inclusion of classical codes among quantum codes and the feedback converse.

For the error bound, fix the transmitted message and a survivor set of size $k$. Let $Z$ count the other codewords with the same surviving coordinates. For the random ensemble, $\E[Z]=(M-1)d^{-k}$. Uniform guessing among compatible messages gives conditional error $Z/(1+Z)$. For every nonnegative integer $Z$,
\begin{equation}
 \frac{Z}{1+Z}\leq\min\{1,Z/2\}.
 \label{eq:collision-error}
\end{equation}
The expected error is at most both one and $\E[Z]/2$. Averaging over patterns proves \eqref{eq:random-error-bound}, and some deterministic codebook has error no greater than this average. The codebooks establishing the two assertions need not be the same.
\end{proof}

The constant-factor comparison in \eqref{eq:constant-success-bracket} will determine the strong-converse exponent. It does not determine the error exponent below capacity, since subtracting a constant-factor success bound from one destroys its precision near unit success. We use the separate error bound \eqref{eq:random-error-bound} for that question.

For the normal approximation, it is useful to obtain a maximum-error statement without an initial shared seed. Ordinary removal of the worst messages can multiply the entire error bound by a constant, which would change the Gaussian quantile. Instead, we bound the low-survivor probability uniformly for every message and remove messages only to control collisions on the remaining patterns.

\begin{lemma}[A thresholded maximum-error bound]\label{lem:threshold-expurgation}
For every positive integer $M$ and every real threshold $r$, there is a deterministic classical no-feedback code with $M$ messages and maximum error at most
\begin{equation}
 \min\!\left\{1,\Pr(K<r)
 +(2M-1)\sum_{\substack{0\leq k\leq n\\k\geq r}}b_kd^{-k}\right\}.
 \label{eq:threshold-max-error}
\end{equation}
\end{lemma}
\begin{proof}
Draw $2M$ independent uniform words $x_1,\ldots,x_{2M}$. For each message $i$, define its collision cost on patterns with at least $r$ survivors by
\begin{equation}
 A_i\coloneqq\sum_{j\neq i}\ 
 \sum_{S:\,|S|\geq r}w_S\,
 \one\{(x_i)_S=(x_j)_S\}.
 \label{eq:threshold-collision-cost}
\end{equation}
For two independently drawn words, agreement on $k$ specified coordinates has probability $d^{-k}$. Thus the expected average of the $2M$ costs is
\begin{equation}
 \E\!\left[\frac1{2M}\sum_{i=1}^{2M}A_i\right]
 =(2M-1)\sum_{k\geq r}b_kd^{-k}.
 \label{eq:threshold-cost-expectation}
\end{equation}
Choose a codebook whose average cost does not exceed this expectation, and retain the $M$ messages with smallest costs. Each retained cost is at most twice the original average: otherwise more than $M$ of the original $2M$ costs would exceed twice that average. Deleting messages can only decrease the remaining collision costs.

On a pattern with $K<r$, bound each retained message's error by one. On every other pattern, use \eqref{eq:collision-error}; its error is at most half the number of remaining colliding messages. The maximum error is therefore at most $\Pr(K<r)+\frac12\max_i A_i$, with the costs taken before deletion. Equation~\eqref{eq:threshold-cost-expectation} proves the stated bound. The low-survivor probability has not been multiplied by the factor from removing messages.
\end{proof}

\begin{corollary}[The same normal approximation without feedback]\label{cor:normal-nf}
Fix $p,\eps\in(0,1)$ and $d\geq2$. Under either the average-error or the maximum-error criterion, and without shared randomness,
\begin{equation}
 C_{n,\eps}^{\nf}
 =n(1-p)\log_2d
 +\sqrt{np(1-p)}\,\log_2d\,\Phi^{-1}(\eps)+O(1).
 \label{eq:normal-nf}
\end{equation}
As in Corollary~\ref{cor:normal}, the remainder is bounded by a constant independent of every positive blocklength $n$. The same expansion holds for the classical $d$-ary erasure channel without feedback.
\end{corollary}
\begin{proof}
The feedback capacity is an upper bound under either error criterion, so Corollary~\ref{cor:normal} supplies the upper estimate. For achievability, apply Lemma~\ref{lem:threshold-expurgation} with $r=\log_d M$. The high-survivor term satisfies
\begin{equation}
 \begin{aligned}
 (2M-1)\sum_{k\geq r}b_kd^{-k}
 &\leq2\max_k b_k\sum_{k=\lceil r\rceil}^{\infty}d^{r-k}\\
 &\leq\frac{2d}{d-1}\max_k b_k
 \leq\frac{2d\,u_p}{(d-1)\sqrt n},
 \end{aligned}
 \label{eq:nf-high-survivor-bound}
\end{equation}
where $u_p$ is defined in \eqref{eq:binomial-largest-atom}. Combining this with \eqref{eq:binomial-normal}, a code exists whose maximum error is at most
\begin{equation}
 \Phi\!\left(\frac{\log_d M-nq}{\sqrt{npq}}\right)
 +\frac{v_p+2d\,u_p/(d-1)}{\sqrt n},
 \qquad q=1-p.
 \label{eq:nf-max-error-normal}
\end{equation}
Set $z_\eps=\Phi^{-1}(\eps)$ and use the positive constant $c_\eps$ from \eqref{eq:quantile-density-lower}. Choose a fixed $B>0$ such that
\begin{equation}
 \frac{c_\eps B}{\sqrt{pq}}>v_p+\frac{2d}{d-1}u_p.
 \label{eq:nf-threshold-backoff}
\end{equation}
For sufficiently large $n$, the message size
$M=\lfloor d^{\,nq+\sqrt{npq}z_\eps-B}\rfloor$ is positive. Its standardized logarithm is at most $z_\eps-B/\sqrt{npq}$. For large $n$ this latter point lies in $[z_\eps-1,z_\eps]$, where the normal density is at least $c_\eps$. The mean value theorem and \eqref{eq:nf-threshold-backoff} make \eqref{eq:nf-max-error-normal} at most $\eps$. Rounding down costs at most one bit, giving the claimed constant-order lower estimate for maximum error, and hence for average error. Enlarging the constant to cover the finitely many smaller $n$ proves the uniform statement.
\end{proof}

The classical binary bounded-remainder expansion under both error criteria was already established in~\cite[Theorem~53]{PPV10}. The first- and second-order terms in \eqref{eq:normal-nf} also agree with the unassisted quantum result of~\cite[Theorem~13]{WFT19}, whose remainder is $O(\log n)$. The present proof gives a $d$-ary classical achievability argument and combines it with the quantum converse to obtain a bounded remainder for unrestricted quantum encodings under both error criteria. The finite-blocklength separation in Sections~\ref{sec:strict}--\ref{sec:tetra} is therefore compatible with identical first- and second-order asymptotics, and even with an $O(1)$ difference between the capacities with and without feedback at a fixed error tolerance.

\subsection{The common strong-converse exponent}
We now keep the rate above capacity fixed and ask how quickly the optimal success probability decays. The exact formula is a sum over survivor fractions $t=k/n$. Each fraction incurs a binomial large-deviation cost and, when too few systems survive, an additional dimension deficit. The exponent is obtained by minimizing the sum of these two costs.

For $t,q\in(0,1)$, define the binary relative entropy
\begin{equation}
 \binD(t\|q)\coloneqq t\log_2\frac tq+(1-t)\log_2\frac{1-t}{1-q},
 \label{eq:binary-relative-entropy}
\end{equation}
with the continuous extension at $t=0,1$. A rate above capacity forces the success probability to vanish. The exponent below quantifies its precise exponential order.

\begin{corollary}[The strong-converse exponent with and without feedback]\label{cor:exponent}
Fix $p\in(0,1)$, $d\geq2$, and $R>(1-p)\log_2d$. Set $q\coloneqq1-p$ and $M_n\coloneqq\lceil2^{nR}\rceil$. Then
\begin{equation}
 \lim_{n\to\infty}-\frac1n\log_2\succp^{*,\fb}(n,M_n)
 =\min_{0\leq t\leq1}\left\{\binD(t\|q)+(R-t\log_2d)_+\right\}.
 \label{eq:exponent-variational}
\end{equation}
Writing $t_*\coloneqq dq/(p+dq)$, the value is
\begin{equation}
 E_{\mathrm{sc}}^{\fb}(R)=
 \begin{cases}
 \binD(R/\log_2d\|q),&q\log_2d<R\leq t_*\log_2d,\\[2pt]
 R-\log_2(p+dq),&R\geq t_*\log_2d.
 \end{cases}
 \label{eq:exponent-explicit}
\end{equation}
For average success probability, the same limit and value hold without feedback, both for unrestricted quantum codes and for classical $d$-ary erasure codes:
\begin{equation}
 E_{\mathrm{sc}}^{\nf}(R)=E_{\mathrm{sc}}^{\cl,\nf}(R)
 =E_{\mathrm{sc}}^{\fb}(R).
 \label{eq:sc-nf-equality}
\end{equation}
The no-feedback assertion here concerns average success; it does not assert an equality of the finite-blocklength average and worst-message success probabilities.
\end{corollary}
\begin{proof}
For a type $t=k/n$, the elementary binomial type bounds give
\begin{equation}
 \frac1{n+1}2^{-n\binD(t\|q)}
 \leq\Pr\!\left(K=k\right)\leq2^{-n\binD(t\|q)}.
 \label{eq:binomial-type-bounds}
\end{equation}
For example, these follow by multiplying
$(n+1)^{-1}2^{nh_2(t)}\leq\binom nk\leq2^{nh_2(t)}$
by $q^k(1-q)^{n-k}$, where $h_2(t)$ is the binary entropy. Set $R_n\coloneqq n^{-1}\log_2M_n$. The factor in the exact success formula is
\begin{equation}
 \min\{1,d^k/M_n\}=2^{-n(R_n-t\log_2d)_+}.
 \label{eq:success-type-factor}
\end{equation}
Let $a_n$ be the minimum of $\binD(t\|q)+(R_n-t\log_2d)_+$ over $t\in\{0,1/n,\ldots,1\}$. The summand at a minimizing type is at least $2^{-na_n}/(n+1)$, whereas each of the $n+1$ summands is at most $2^{-na_n}$. Therefore,
\begin{equation}
 \frac{2^{-na_n}}{n+1}
 \leq\succp^{*,\fb}(n,M_n)
 \leq(n+1)2^{-na_n}.
 \label{eq:exponent-finite-bracket}
\end{equation}
Taking negative normalized logarithms places the exponent within $\log_2(n+1)/n$ of $a_n$. The objectives converge uniformly when $R_n\to R$, since the positive-part function is Lipschitz with constant one. The limiting objective is continuous on the compact interval $[0,1]$, and each point in that interval is within $1/n$ of a type. Consequently, the grid minima converge to the continuous minimum, proving \eqref{eq:exponent-variational}.

To evaluate the minimum, observe that
\begin{equation}
 \binD(t\|q)-t\log_2d
 =\binD(t\|t_*)-\log_2(p+dq).
 \label{eq:exponential-tilt-identity}
\end{equation}
Indeed, $t_*/q=d/(p+dq)$ and $(1-t_*)/p=1/(p+dq)$, so expansion of the logarithms proves the identity. The expression on the left is minimized at $t_*$. On the interval $t\geq R/\log_2d$ when this interval is nonempty, the objective in \eqref{eq:exponent-variational} is $\binD(t\|q)$, which increases for $t>q$. On $t\leq R/\log_2d$, it is $R+\binD(t\|q)-t\log_2d$, which decreases up to $t_*$ and increases after $t_*$. The minimum is therefore at $R/\log_2d$ when $R/\log_2d\leq t_*$, and at $t_*$ otherwise. Equation~\eqref{eq:exponential-tilt-identity} gives \eqref{eq:exponent-explicit}. Both expressions agree at the boundary.

Finally, the constant-factor comparison in \eqref{eq:constant-success-bracket} applies to the optimal average success probabilities of both classical and quantum codes without feedback. Each is at most the feedback-assisted value and at least $1-e^{-1}$ times that value. Their negative normalized logarithms therefore differ from the feedback expression by at most $-n^{-1}\log_2(1-e^{-1})$, which tends to zero. This proves \eqref{eq:sc-nf-equality} without assuming that a no-feedback code attains the exact finite-blocklength feedback formula.
\end{proof}

Figure~\ref{fig:strong-converse} shows the exponent in Corollary~\ref{cor:exponent} for $d=2$ and $p=1/2$, together with $-n^{-1}\log_2\succp^{*,\fb}(n,\lceil2^{nR}\rceil)$ at $n=20,100,500$. The limiting curve becomes linear at $R=t_*\log_2d=2/3$. The finite-blocklength curves here and in Figure~\ref{fig:feedback-reliability} connect sampled rates, retaining the integer message size. Only feedback-assisted values are plotted at finite $n$: equality of the limiting exponents does not imply equality of finite-blocklength performance with and without feedback.

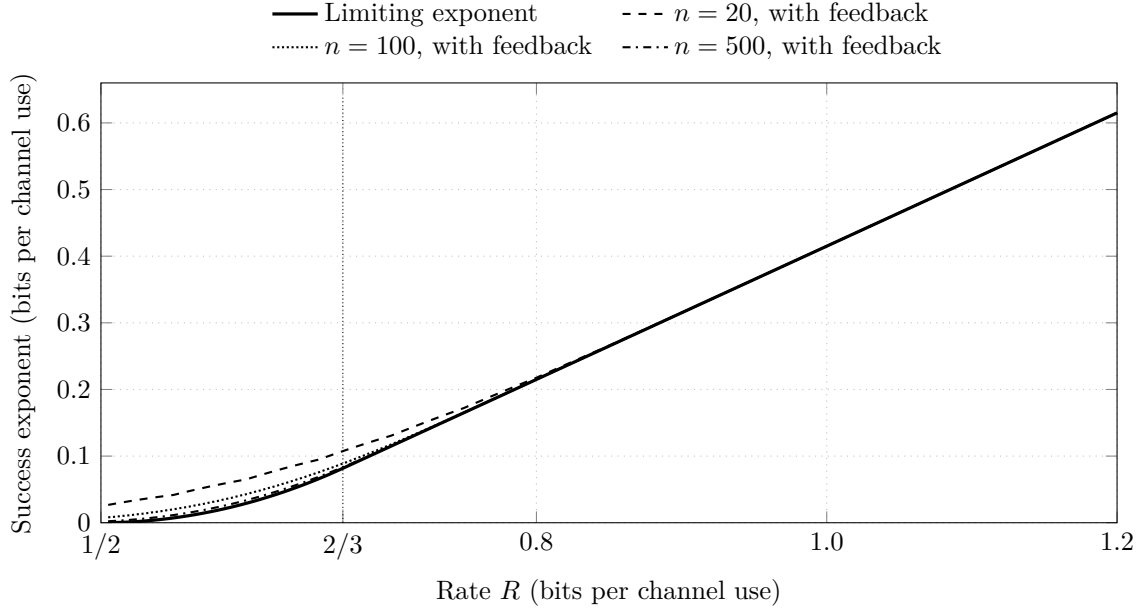
\begin{figure}[t]
 \centering
 \input{figures/strong_converse.tex}
 \caption{Above-capacity strong-converse exponent from Corollary~\ref{cor:exponent}, for $d=2$ and $p=1/2$. The solid curve is the common limiting exponent for optimal average success with or without feedback. The other curves evaluate $-n^{-1}\log_2\succp^{*,\fb}(n,\lceil2^{nR}\rceil)$ using the exact feedback formula. For $1/2<R\leq2/3$ the limiting exponent is $D_{\mathrm{bin}}(R\|1/2)$; for $R\geq2/3$ it is $R-\log_2(3/2)$. All rates plotted satisfy $R>1/2$.}
 \label{fig:strong-converse}
\end{figure}

\subsection{Error exponents below capacity}\label{subsec:reliability}
Above capacity, the strong-converse exponent measures the decay of the success probability. Below capacity, the relevant quantity is instead the decay of the error probability. These are different limits. In particular, the constant-factor success comparison in \eqref{eq:constant-success-bracket} does not imply equality of the below-capacity error exponents.

For $\mathsf a\in\{\fb,\nf\}$, write
$\eps_{\mathsf a}^*(n,M)\coloneqq1-\succp^{*,\mathsf a}(n,M)$ for the optimal average error. Fix a positive rate
\begin{equation}
 0<R<(1-p)\log_2d,
 \qquad M_n\coloneqq\lceil2^{nR}\rceil.
 \label{eq:reliability-rate-range}
\end{equation}
An error exponent is the limit of $-n^{-1}\log_2\eps_{\mathsf a}^*(n,M_n)$, when it exists. We use lower and upper limits when equality has not been proved. The restriction $R>0$ avoids the different endpoint convention $M_n=1$ at $R=0$, for which the error is identically zero.

\begin{corollary}[Exact feedback-assisted error exponent]\label{cor:reliability-fb}
For $p\in(0,1)$, $d\geq2$, and the rates in \eqref{eq:reliability-rate-range},
\begin{equation}
 E_{\mathrm{rel}}^{\fb}(R)
 \coloneqq\lim_{n\to\infty}-\frac1n\log_2\eps_{\fb}^*(n,M_n)
 =\binD\!\left(\frac{R}{\log_2d}\middle\|1-p\right).
 \label{eq:fb-reliability-exact}
\end{equation}
The same exponent holds for maximum error under the initial-feedback-seed convention of Theorem~\ref{thm:feedback}.
\end{corollary}
\begin{proof}
Set $q=1-p$, $a=R/\log_2d\in(0,q)$, and $r_n=\log_d M_n$. Then $r_n/n\to a$. From the exact error formula,
\begin{equation}
 \eps_{\fb}^*(n,M_n)
 =\sum_{k<r_n}b_k(1-d^{k-r_n})
 \leq\Pr(K<r_n).
 \label{eq:reliability-tail-upper}
\end{equation}
Let $j_n=\lceil r_n\rceil-1$. For all sufficiently large $n$, $0\leq j_n<nq$. The binary relative entropy $\binD(t\|q)$ decreases on $0\leq t<q$. Applying the upper type bound in \eqref{eq:binomial-type-bounds} to each of the at most $n+1$ terms therefore gives
\begin{equation}
 \Pr(K<r_n)\leq(n+1)2^{-n\binD(j_n/n\|q)}.
 \label{eq:reliability-tail-type}
\end{equation}
For the reverse exponential estimate, choose $k_n=\lfloor r_n\rfloor-1$, which lies in $\{0,\ldots,n\}$ for sufficiently large $n$. Since $r_n-k_n\geq1$, one term of the exact error yields
\begin{equation}
 \eps_{\fb}^*(n,M_n)
 \geq(1-d^{-1})b_{k_n}
 \geq\frac{1-d^{-1}}{n+1}\,2^{-n\binD(k_n/n\|q)}.
 \label{eq:reliability-single-type}
\end{equation}
Both $j_n/n$ and $k_n/n$ tend to $a$. Taking negative normalized logarithms of \eqref{eq:reliability-tail-type} and \eqref{eq:reliability-single-type}, and using continuity of the binary relative entropy, proves the limit. Equality of the finite-blocklength average and maximum errors in the stated feedback model gives the final assertion.
\end{proof}

Figure~\ref{fig:feedback-reliability} illustrates Corollary~\ref{cor:reliability-fb} for $d=2$ and $p=1/2$. It compares the exact limiting error exponent with the negative normalized logarithm of the optimal feedback error at finite blocklength. Unlike Figure~\ref{fig:strong-converse}, this figure concerns the decay of error below capacity, not the decay of success above capacity. The behavior near $R=0$ also reflects rounding of $\lceil2^{nR}\rceil$; the theorem takes $n\to\infty$ at a fixed positive rate.

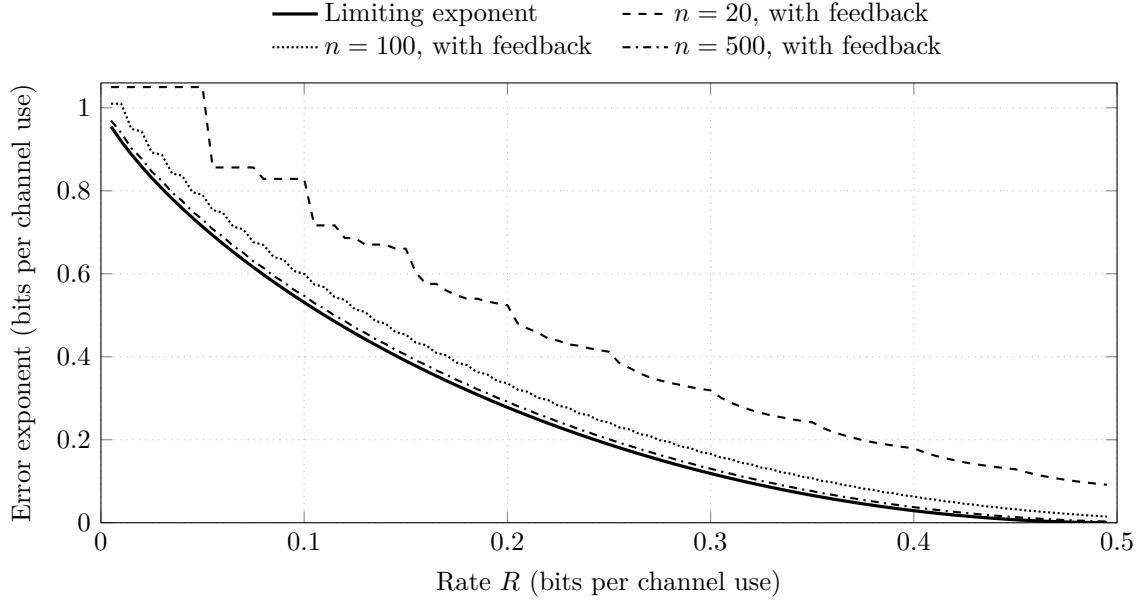
\begin{figure}[t]
 \centering
 \input{figures/feedback_error_exponent.tex}
 \caption{Below-capacity feedback-assisted error exponent from Corollary~\ref{cor:reliability-fb}, for $d=2$ and $p=1/2$. The solid curve is $D_{\mathrm{bin}}(R\|1/2)$ on $0<R<1/2$. The other curves are $-n^{-1}\log_2\eps_\fb^*(n,\lceil2^{nR}\rceil)$ computed from the exact finite-blocklength error. No value at $R=0$ is assigned: setting $R=0$ would give a single-message code and zero error, whereas the limiting exponent as $R\downarrow0$ is one. The no-feedback reliability bounds are compared separately in Figure~\ref{fig:reliability}.}
 \label{fig:feedback-reliability}
\end{figure}

Without feedback, classical random coding supplies an achievable exponent, whereas the feedback formula supplies a converse even for entangled quantum codewords. The two bounds coincide above a critical rate. At lower rates, we also give the classical expurgated lower bound, obtained by removing messages with unusually large pairwise confusion probabilities. These are specializations of the random-coding and expurgation methods of Gallager~\cite{Gallager65}; the short derivations below identify exactly what they imply for the present quantum problem.

Set $q=1-p$ and define
\begin{equation}
 t_{\mathrm{crit}}\coloneqq\frac{q}{q+dp},
 \qquad R_{\mathrm{crit}}\coloneqq t_{\mathrm{crit}}\log_2d.
 \label{eq:critical-rate}
\end{equation}
Since $0<t_{\mathrm{crit}}<q$, this critical rate lies strictly between zero and capacity. Define the random-coding exponent
\begin{equation}
 \begin{aligned}
 E_{\mathrm r}(R)
 &\coloneqq\min_{0\leq t\leq1}
 \left\{\binD(t\|q)+(t\log_2d-R)_+\right\}\\
 &=\begin{cases}
 -\log_2(p+q/d)-R,&0<R\leq R_{\mathrm{crit}},\\[2pt]
 \binD(R/\log_2d\|q),&R_{\mathrm{crit}}\leq R<q\log_2d,
 \end{cases}
 \end{aligned}
 \label{eq:random-coding-exponent}
\end{equation}
and the expurgated exponent
\begin{equation}
 E_{\mathrm{ex}}(R)\coloneqq
 \sup_{s\geq1}\left\{
 -s\log_2\!\left[\frac1d+\left(1-\frac1d\right)p^{1/s}\right]-sR
 \right\}.
 \label{eq:expurgated-exponent}
\end{equation}
The parameter $s$ in \eqref{eq:expurgated-exponent} is a real optimization variable, unrelated to the survivor set $S$ or the success function $s_{n,d,p}$.

\begin{proposition}[No-feedback error exponents]\label{prop:reliability-nf}
For the rates in \eqref{eq:reliability-rate-range}, define
\begin{equation}
 \begin{aligned}
 \underline E_{\mathrm{rel}}^{\nf}(R)
 &\coloneqq\liminf_{n\to\infty}-\frac1n\log_2\eps_{\nf}^*(n,M_n),\\
 \overline E_{\mathrm{rel}}^{\nf}(R)
 &\coloneqq\limsup_{n\to\infty}-\frac1n\log_2\eps_{\nf}^*(n,M_n).
 \end{aligned}
 \label{eq:nf-reliability-definitions}
\end{equation}
Then
\begin{equation}
 \max\{E_{\mathrm r}(R),E_{\mathrm{ex}}(R)\}
 \leq\underline E_{\mathrm{rel}}^{\nf}(R)
 \leq\overline E_{\mathrm{rel}}^{\nf}(R)
 \leq\binD(R/\log_2d\|q).
 \label{eq:nf-reliability-bracket}
\end{equation}
In particular, for $R_{\mathrm{crit}}\leq R<q\log_2d$, the limit exists and equals the feedback-assisted exponent:
\begin{equation}
 E_{\mathrm{rel}}^{\nf}(R)
 =E_{\mathrm{rel}}^{\fb}(R)
 =\binD(R/\log_2d\|q).
 \label{eq:high-rate-reliability-equality}
\end{equation}
The bounds and the exact high-rate formula hold also for maximum error without shared randomness. They hold for classical $d$-ary erasure codes as well as for unrestricted quantum codes without feedback.
\end{proposition}
\begin{proof}
A code without feedback is a special feedback-assisted protocol, so $\eps_{\nf}^*(n,M_n)\geq\eps_{\fb}^*(n,M_n)$. Corollary~\ref{cor:reliability-fb} proves the upper bound in \eqref{eq:nf-reliability-bracket}, including for quantum codewords entangled across all inputs.

For achievability, \eqref{eq:random-error-bound} is at most
\begin{equation}
 \sum_{k=0}^n b_k\min\{1,M_nd^{-k}\}
 =\sum_{k=0}^n b_k\,2^{-n((k/n)\log_2d-R_n)_+},
 \qquad R_n\coloneqq\frac1n\log_2M_n.
 \label{eq:random-error-types}
\end{equation}
The upper type bound and a sum over at most $n+1$ types show that the lower limit of its negative normalized logarithm is at least $E_{\mathrm r}(R)$. The convergence of the grid minima follows exactly as in the proof of Corollary~\ref{cor:exponent}. Since a deterministic classical code attains the error bound, this proves the random-coding lower estimate for the optimal quantum error as well.

To evaluate the minimum, the relevant tilted probability is now $t_{\mathrm{crit}}$, rather than the $t_*$ used for the strong-converse exponent. Expansion of the logarithms gives
\begin{equation}
 \binD(t\|q)+t\log_2d
 =\binD(t\|t_{\mathrm{crit}})-\log_2(p+q/d).
 \label{eq:reliability-tilt}
\end{equation}
On $t\leq R/\log_2d<q$, the objective in \eqref{eq:random-coding-exponent} is $\binD(t\|q)$ and decreases up to the endpoint. On $t\geq R/\log_2d$, it is $\binD(t\|q)+t\log_2d-R$, minimized at $t_{\mathrm{crit}}$ when that point is in the interval, and otherwise at its left endpoint. This proves the two cases in \eqref{eq:random-coding-exponent}.

The expurgation argument in Appendix~\ref{app:expurgation} gives deterministic classical codes attaining $E_{\mathrm{ex}}(R)$ under maximum error. To pass the random-coding exponent to maximum error as well, begin with a $2M_n$-message code attaining \eqref{eq:random-error-bound}. Keep the $M_n$ messages with smallest individual errors. Their errors, using the old decoder and assigning deleted outcomes to arbitrary retained messages, are at most twice the old average error. The factor two and the change from $M_n$ to $2M_n$ do not affect an error exponent, since their binary logarithms are $O(1)$. This establishes both lower bounds under maximum error without shared randomness. The converse still applies because maximum error is no smaller than average error. Finally, the random-coding and converse bounds agree when $R\geq R_{\mathrm{crit}}$, giving \eqref{eq:high-rate-reliability-equality}.
\end{proof}

For $0<R<R_{\mathrm{crit}}$, \eqref{eq:nf-reliability-bracket} is a pair of bounds, not a claimed exact no-feedback reliability function. In particular, a gap between these bounds does not prove a feedback advantage in the optimal error exponent. Figure~\ref{fig:reliability} illustrates what is and is not determined. The finite-blocklength quantum advantage in Section~\ref{sec:tetra} also does not, by itself, establish an improvement in a limiting error exponent.

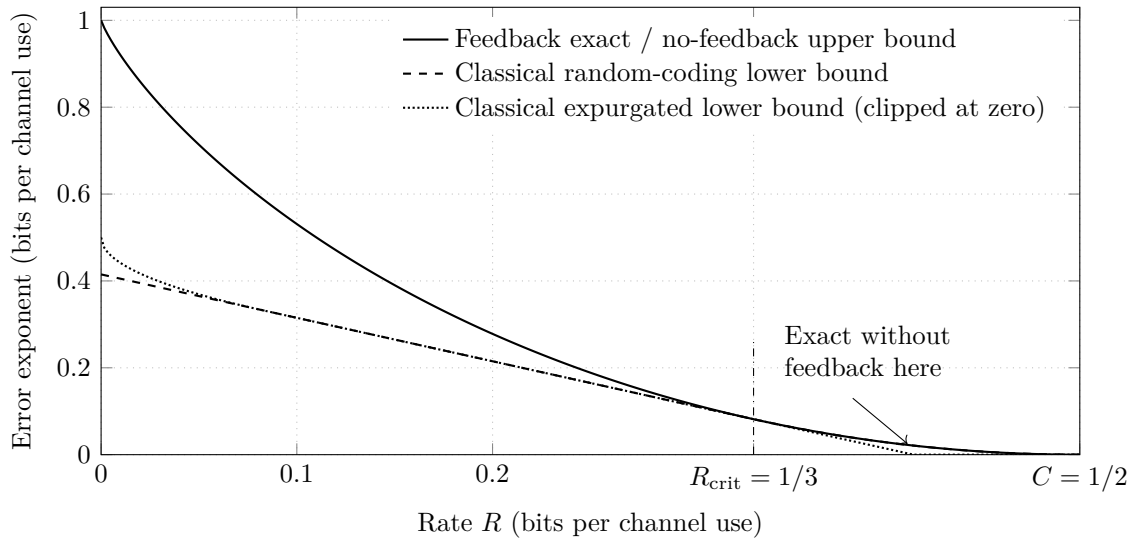
\begin{figure}[t]
 \centering
 \input{figures/reliability_bounds.tex}
 \caption{Below-capacity error exponents for $d=2$ and $p=1/2$. The upper solid curve is the exact feedback-assisted exponent and an upper bound without feedback. The dashed and dotted curves are achievable without feedback by classical random coding and expurgation, respectively. For $R_{\mathrm{crit}}=1/3\leq R<1/2$, the random-coding curve coincides with the upper curve, determining the classical and quantum no-feedback exponents exactly. At lower positive rates only the indicated bounds are asserted. The values at $R=0$ are limiting values as $R\downarrow0$, not exponents for a single-message code.}
 \label{fig:reliability}
\end{figure}

\section{Conclusion}\label{sec:conclusion}

\subsection{Summary of results}
We determined the exact fixed-blocklength success probability for classical communication over a quantum erasure channel assisted by classical feedback and without initial shared entanglement. The converse allows adaptive encoders, arbitrary quantum memories, and arbitrary receiver instruments. Its central step is a noiseless communication dimension bound. Fixing the classical controls yields normalized sender states, while the receiver's complete instrument tree supplies a POVM containing the transcript probabilities. This permits a dimension estimate without assuming separability during the protocol or a uniform posterior message distribution.

For the erasure channel, a fixed pattern reduces the problem to its surviving noiseless transfers. The repeat-until-received protocol using base-$d$ digits attains the resulting bound for every message size. We also evaluated the associated message-size formula. Classical no-feedback codes show that the bounded-remainder normal approximation and the average-success strong-converse exponent coincide with their feedback-assisted counterparts. Below capacity, the feedback-assisted error exponent is exact at every positive rate; the no-feedback exponent agrees with it above the critical rate, while random-coding and expurgated lower bounds apply below that rate.

Without feedback, a relative-majorization construction preserves an entire family of sector-uniform auxiliary states and gives an exact evaluation of a hypothesis-testing converse. We recovered the numerical converse of~\cite{Rosati24}, but proved that it is not generally an exact operational formula. Four messages and four qubit transmissions provide a strict counterexample to universal attainability. For these same parameters, an explicit tetrahedral product-state code strictly improves on every classical binary erasure code. The three coding models therefore have genuinely different finite-blocklength performances.

\subsection{Directions for future research}
The exact unassisted quantum success probability remains to be determined even for the four-message, four-qubit example. The tetrahedral code is an explicit lower bound, not a proof of optimality. It would be useful to determine whether another product-state ensemble improves it and whether entangled codewords yield a further advantage.

The lower and upper no-feedback reliability bounds in Proposition~\ref{prop:reliability-nf} need not coincide below the critical rate. Improving those bounds and determining whether nonorthogonal or entangled codewords improve the asymptotic error exponent are separate questions from the finite-blocklength separation proved here.

More generally, equality in the sectorwise dimension bound imposes simultaneous distinguishability and purity conditions on many reduced code states. A systematic characterization of these conditions could identify additional parameter regimes with exact unassisted formulas, or provide quantitative gaps rather than the compactness-based strict inequality proved here.

The feedback result also suggests studying restrictions on the backward communication. The attaining protocol needs only acknowledgments after each forward use, together with an initial seed when maximum-error symmetrization is required. Delayed or noisy acknowledgments change that structure. Understanding which restrictions preserve an exact formula would clarify the role of feedback beyond the unrestricted setting considered here.

\section*{Acknowledgments}
The author acknowledges helpful discussions with Hemant Mishra and support from the Cornell University School of Electrical and Computer Engineering.

\section*{Statement on AI-assisted preparation}
ChatGPT (OpenAI) Pro 6 (Astra) was used extensively in developing this manuscript, including for mathematical exploration, formulation and checking of proof arguments, drafting and revision, bibliographic checks, and preparation of LaTeX, TikZ, and numerical-checking code. The research problem and successive revisions were directed by the author, who supplied detailed comments on the mathematical arguments and their presentation.

AI-assisted mathematical and editorial review passes were followed by further revisions. These were internal self-reviews, not independent external peer review or formal verification. Numerical calculations were used to check finite-dimensional identities and examples, not as substitutes for proofs. Responsibility for the mathematical statements, citations, and presentation rests with the author.
\bibliographystyle{alphaurl}
\bibliography{references}

\appendix
\section{Symmetry reduction of the fully optimized hypothesis-testing converse}\label{app:symmetry}
This appendix proves the lower bound in Proposition~\ref{prop:full-meta}. The point is to reduce the auxiliary state by data processing on the joint classical--quantum system, rather than by assuming a convexity property of the hypothesis-testing divergence.

Let $\mathbb Z_d$ denote the integers modulo $d$, and label the computational-basis input words by $x\in\mathbb Z_d^n$. Their uniform output ensemble is
\begin{equation}
 \Theta_{XB^n}\coloneqq\frac1{d^n}\sum_{x\in\mathbb Z_d^n}
 \proj x_X\otimes\bigoplus_{S\subseteq[n]}w_S\proj{x_S},
 \qquad\pi_X\coloneqq I_X/d^n.
 \label{eq:uniform-basis-ensemble}
\end{equation}
Here $x_S$ is the restriction of the word to surviving positions. We show that
\begin{equation}
 \inf_\sigma\Htest^\eps(\Theta_{XB^n}\|\pi_X\otimes\sigma)
 =\inf_q\Htest^\eps(\omega_0\|\sigma_q),
 \qquad\omega_0\coloneqq(\cE^p)^{\otimes n}(\left(\proj0\right)^{\otimes n}).
 \label{eq:symmetry-target}
\end{equation}

First, pinch the output into its erasure sectors. This leaves $\Theta$ unchanged and converts an arbitrary auxiliary $\sigma$ into its sector-diagonal part. Next average over powers of the local phase unitaries
\begin{equation}
 Z\ket j\coloneqq e^{2\pi i j/d}\ket j,
 \qquad Z\ket e\coloneqq\ket e,
 \label{eq:phase-unitary}
\end{equation}
independently in each coordinate. These channels also fix $\Theta$ and remove every off-diagonal entry within each surviving computational basis. The auxiliary is now diagonal within each sector. Denote it by $\sigma'$.

Let $U_g$ be the tensor product of cyclic shifts $\ket{j_i}\mapsto\ket{j_i+g_i}$ on the nonerased part of each output, fixing the erasure vector. Let $V_g$ shift the classical label by $x\mapsto x+g$. For every $g\in\mathbb Z_d^n$,
\begin{equation}
 (V_g\otimes U_g)\Theta(V_g\otimes U_g)^\dagger=\Theta.
 \label{eq:joint-covariance}
\end{equation}
Average these joint unitaries uniformly. Since $V_g\pi_XV_g^\dagger=\pi_X$, their action on the second hypothesis is
\begin{equation}
 \pi_X\otimes\sigma'
 \longmapsto\pi_X\otimes\frac1{d^n}\sum_g U_g\sigma'U_g^\dagger
 =\pi_X\otimes\sigma_q,
 \label{eq:aux-twirl}
\end{equation}
where $q_S=\Tr[P_S\sigma]$ is the unchanged total weight of sector $S$. The last equality follows because independent cyclic shifts act transitively on all basis words within each sector. The first hypothesis remains $\Theta$ throughout. Data processing therefore proves that for every $\sigma$ there is a $q$ such that
\begin{equation}
 \Htest^\eps(\Theta\|\pi_X\otimes\sigma)
 \geq\Htest^\eps(\Theta\|\pi_X\otimes\sigma_q).
 \label{eq:symmetry-dpi}
\end{equation}
Since the states $\sigma_q$ are themselves admissible auxiliaries, it follows in both directions that the infimum over all $\sigma$ equals the infimum restricted to this family.

Finally, apply the controlled unitary
\begin{equation}
 W\coloneqq\sum_{x\in\mathbb Z_d^n}\proj x_X\otimes U_{-x}.
 \label{eq:controlled-inverse}
\end{equation}
It takes $\Theta$ to $\pi_X\otimes\omega_0$. Every $\sigma_q$ is uniform within its sectors, so $W$ fixes $\pi_X\otimes\sigma_q$. Unitary invariance and removal of the common independent factor $\pi_X$ give
\begin{equation}
 \Htest^\eps(\Theta\|\pi_X\otimes\sigma_q)
 =\Htest^\eps(\omega_0\|\sigma_q).
 \label{eq:controlled-equality}
\end{equation}
This proves \eqref{eq:symmetry-target}. Theorem~\ref{thm:testing-value} evaluates its right-hand side as $\log_2R_\eps$. Therefore the uniform computational-basis ensemble attains the upper bound on \eqref{eq:fully-optimized-meta}, completing the proof of Proposition~\ref{prop:full-meta}.

\section{The classical expurgated error bound}\label{app:expurgation}
For completeness, we derive the lower bound \eqref{eq:expurgated-exponent} using the classical expurgation method~\cite{Gallager65}. This supplies a low-rate achievability bound for the quantum erasure channel using only computational-basis inputs. The argument bounds every retained message's error, so no shared randomness is required.

Fix $s\geq1$. For two words $x_i,x_j\in\{0,\ldots,d-1\}^n$, let $h_{ij}$ be their Hamming distance, the number of coordinates at which they differ. Their surviving projections agree exactly when all of these differing positions are erased, an event of probability $p^{h_{ij}}$. Uniform guessing among all compatible messages therefore gives individual error at most
\begin{equation}
 \eps_i\leq\frac12\sum_{j\neq i}p^{h_{ij}},
 \label{eq:expurgated-pair-error}
\end{equation}
since $z/(1+z)\leq z/2$ when $z$ other messages are compatible. The right side is allowed to exceed one; it is an upper bound, not a probability identity.

Draw $2M$ independent uniform words and set
\begin{equation}
 A_i^{(s)}\coloneqq\sum_{j\neq i}p^{h_{ij}/s},
 \qquad
 a_s\coloneqq\frac1d+\left(1-\frac1d\right)p^{1/s}.
 \label{eq:expurgated-cost}
\end{equation}
Two independent symbols agree with probability $1/d$. Independence across coordinates consequently gives
\begin{equation}
 \E\!\left[\frac1{2M}\sum_{i=1}^{2M}A_i^{(s)}\right]
 =(2M-1)a_s^n.
 \label{eq:expurgated-mean}
\end{equation}
Choose a codebook whose average cost does not exceed this value and retain its $M$ smallest-cost messages. As in Lemma~\ref{lem:threshold-expurgation}, every retained cost is at most $2(2M-1)a_s^n$, and deleting the other messages can only decrease the cost. For nonnegative numbers $(z_j)_j$ and $s\geq1$, $\sum_j z_j^s\leq(\sum_jz_j)^s$. Applying this inequality to $z_j=p^{h_{ij}/s}$ in \eqref{eq:expurgated-pair-error} gives a deterministic code satisfying
\begin{equation}
 \eps_{\max}\leq\frac12\left[2(2M-1)a_s^n\right]^s.
 \label{eq:expurgated-finite}
\end{equation}
For $M=M_n=\lceil2^{nR}\rceil$ and a fixed $s$, the lower limit of the negative normalized logarithm of this bound is $-s\log_2a_s-sR$. Since the optimal error is no larger than the error of the constructed code, this lower estimate holds for its exponent for every fixed $s\geq1$. Taking the supremum over $s$ proves \eqref{eq:expurgated-exponent}. No interchange of a limit and a varying choice of $s$ is needed.

\end{document}

%% file: figures/noiseless_protocol.tex
\begin{tikzpicture}[x=1cm,y=1cm,font=\small,>=Latex]
  \node[anchor=east] at (-.1,2.5) {Alice};
  \node[anchor=east] at (-.1,0) {Bob};
  \node[op] (a1) at (1,2.5) {$\mathcal A_1^m$};
  \node[op] (a2) at (4.8,2.5) {$\mathcal A_2^{m,f_1}$};
  \node[op] (a3) at (8.6,2.5) {$\mathcal A_3^{m,f_1f_2}$};
  \node[op] (b1) at (2.7,0) {$\mathcal B_{f_1}$};
  \node[op] (b2) at (6.5,0) {$\mathcal B_{f_2|f_1}$};
  \node[op] (b3) at (10.5,0) {$(\Lambda_{\widehat m|f})_{\widehat m}$};
  \draw[qwire] (-.1,2.5)--(a1.west) node[midway,above] {$A_0$};
  \draw[qwire] (a1.east)--(a2.west) node[midway,above] {$A_1$};
  \draw[qwire] (a2.east)--(a3.west) node[midway,above] {$A_2$};
  \draw[qwire] (a3.east)--(11.5,2.5) node[midway,above] {$A_3$};
  \node[anchor=west] at (11.5,2.5) {$\Tr$};
  \draw[qwire] (-.1,0)--(b1.west) node[midway,below] {$B_0$};
  \draw[qwire] (b1.east)--(b2.west) node[midway,below] {$B_1$};
  \draw[qwire] (b2.east)--(b3.west) node[midway,below] {$B_2$};
  \draw[cwire] (b3.east)--(12.2,0) node[right] {$\widehat m$};
  \draw[qwire] (a1.south)--(b1.north) node[midway,left=2pt] {$Q_1$};
  \draw[qwire] (a2.south)--(b2.north) node[midway,left=2pt] {$Q_2$};
  \draw[qwire] (a3.south)--(b3.north) node[midway,left=2pt] {$Q_3$};
  \draw[cwire] (b1.north east)--(a2.south west) node[midway,right=2pt] {$F_1$};
  \draw[cwire] (b2.north east)--(a3.south west) node[midway,right=2pt] {$F_2$};
  \draw[->] (1,-1.0)--(10.7,-1.0) node[midway,below] {time};
  \node[align=center] at (6.1,3.55) {No initial shared entanglement; all forward transfers are noiseless.};
\end{tikzpicture}

%% file: figures/fixed_transcript.tex
\begin{tikzpicture}[x=1cm,y=1cm,font=\small,>=Latex]
\node[anchor=west,font=\bfseries\small] at (-.1,3.85)
 {Prescribe $m$ and $f=(f_1,f_2)$; do not condition on observing $f$.};
\node[op,minimum width=2.9cm,minimum height=2cm,align=center]
 (prep) at (2.55,1.8) {Alice's channels\\with prescribed $m,f$};
\draw[qwire] (-.15,1.8)--(prep.west) node[pos=.30,above=4pt] {$A_0$};
\node[draw,rounded corners,minimum width=3.25cm,minimum height=2cm,align=center]
 (meas) at (9.05,1.8) {Bob's complete instrument\\sequence and decoder};
\draw[qwire] ([yshift=.65cm]prep.east)--([yshift=.65cm]meas.west) node[pos=.5,above=2pt] {$Q_1$};
\draw[qwire] (prep.east)--(meas.west) node[pos=.5,above=1pt,fill=white,inner sep=1pt] {$Q_2$};
\draw[qwire] ([yshift=-.65cm]prep.east)--([yshift=-.65cm]meas.west) node[pos=.5,below=2pt] {$Q_3$};
\draw[cwire] (meas.east)--(12.1,1.8) node[right] {$(f,\widehat m)$};
\node[align=center] at (2.55,-.30)
 {$\rho^{m,f}_{Q_1Q_2Q_3}\geq0$\\[6pt]$\Tr[\rho^{m,f}_{Q_1Q_2Q_3}]=1$};
\node[align=center] at (9.05,-.30)
 {$T_{f,\widehat m}\geq0$\\[6pt]$\displaystyle\sum_{f,\widehat m}T_{f,\widehat m}=I_{Q_1Q_2Q_3}$};
\node[draw,rounded corners,inner sep=7pt,align=center] at (6.2,-2.05)
 {$\Pr\!\left(f,\widehat m|m\right)=\Tr[T_{f,\widehat m}\rho^{m,f}]$\\
 The outcome probability is in the effect, not in the prepared state.};
\end{tikzpicture}

%% file: figures/fixed_pattern.tex
\begin{tikzpicture}[x=1cm,y=1cm,font=\small,>=Latex]
  \node[anchor=west,font=\bfseries\small] at (-.3,3.4) {A fixed erasure pattern: $S=\{1,3\}$};
  \node[anchor=east] at (-.45,2.25) {Alice};
  \node[anchor=east] at (-.45,0) {Bob};
  \draw[qwire] (-.35,2.25)--(11.5,2.25) node[right] {memory};
  \draw[qwire] (-.35,0)--(11.5,0) node[right] {memory};
  \foreach \x/\i in {1/1,4/2,7/3,10/4} {
    \node[op,minimum width=.9cm] (a\i) at (\x,2.25) {$\mathcal A_\i$};
    \node[op,minimum width=.9cm] (b\i) at (\x,0) {$\mathcal B_\i$};
  }
  \foreach \i in {1,3} {
    \draw[qwire] (a\i.south)--(b\i.north) node[midway,left] {$Q_\i$};
  }
  \foreach \x/\i in {4/2,10/4} {
    \draw[qwire] (a\i.south)--(\x,1.45);
    \node[draw,inner sep=3pt] at (\x,1.18) {$\Tr_{Q_\i}$};
    \draw[qwire] (\x,.8)--(b\i.north);
    \node[fill=white,inner sep=2pt] at (\x,.74) {$|e\rangle$};
  }
  \foreach \i/\j in {1/2,2/3,3/4} {
    \draw[cwire] (b\i.north east)--(a\j.south west)
      node[pos=.5,fill=white,inner sep=1pt] {$F_\i$};
  }
  \node[align=center] at (5.6,-1.05)
    {Erased systems are discarded locally; their flags carry no message.\\
     Only two quantum systems cross from Alice to Bob.};
\end{tikzpicture}

%% file: figures/residue_tree.tex
\begin{tikzpicture}[x=1cm,y=1cm,font=\small,>=Latex,
  residue/.style={draw,rounded corners,fill=white,inner sep=5pt}]
  \node[residue] (r) at (6,3.4) {$\{0,1,2,3,4,5,6\}$};
  \node[residue] (e) at (3,1.8) {$\{0,2,4,6\}$};
  \node[residue] (o) at (9,1.8) {$\{1,3,5\}$};
  \node[residue] (a) at (1.5,0) {$\{0,4\}$};
  \node[residue] (b) at (4.5,0) {$\{2,6\}$};
  \node[residue] (c) at (7.5,0) {$\{1,5\}$};
  \node[residue] (d) at (10.5,0) {$\{3\}$};
  \draw[qwire] (r)--(e) node[midway,above left] {digit $0$};
  \draw[qwire] (r)--(o) node[midway,above right] {digit $1$};
  \draw[qwire] (e)--(a) node[midway,left] {$0$};
  \draw[qwire] (e)--(b) node[midway,right] {$1$};
  \draw[qwire] (o)--(c) node[midway,left] {$0$};
  \draw[qwire] (o)--(d) node[midway,right] {$1$};
  \draw[->,dashed] (e.north west) .. controls +(-1.3,.8) and +(-1.3,-.8) .. node[pos=.5,left=3pt] {erasure} (e.south west);
  \node[anchor=west] at (11.3,3.4) {$k=0$};
  \node[anchor=west] at (11.3,1.8) {$k=1$};
  \node[anchor=west] at (11.3,0) {$k=2$};
  \node[align=center] at (6,-1.1)
    {Auxiliary integers $\ell=m-1$, with digits sent least significant first.\\
     Each successful reception refines the residue class; an erasure does not.};
\end{tikzpicture}

%% file: figures/sector_conversion.tex
\begin{tikzpicture}[x=1cm,y=1cm,font=\small,>=Latex]
  \node[anchor=west,font=\bfseries\small] at (0,3.9) {The same channel fixes every sector-uniform auxiliary state.};
  \node[draw,minimum width=3.4cm,minimum height=.95cm,align=center] (l1) at (2,2.4)
    {$w_S\psi^{\otimes |S|}$\\pure within sector $S$};
  \node[draw,minimum width=3.4cm,minimum height=.95cm,align=center] (l2) at (2,.45)
    {$q_S I_S/D_S$\\uniform within sector $S$};
  \node[op,minimum width=2.25cm,minimum height=.95cm,align=center] (c1) at (6.15,2.4)
    {$\mathcal R_S^\rho$\\random unitaries};
  \node[op,minimum width=2.25cm,minimum height=.95cm,align=center] (c2) at (6.15,.45)
    {$\mathcal R_S^\rho$\\the same channel};
  \node[draw,minimum width=3.25cm,minimum height=.95cm,align=center] (r1) at (10.3,2.4)
    {$w_S\rho_S$\\desired marginal};
  \node[draw,minimum width=3.25cm,minimum height=.95cm,align=center] (r2) at (10.3,.45)
    {$q_S I_S/D_S$\\unchanged};
  \draw[qwire] (l1)--(c1);\draw[qwire] (c1)--(r1);
  \draw[qwire] (l2)--(c2);\draw[qwire] (c2)--(r2);
  \node[align=center] at (6.15,-.9)
    {First measure the erasure pattern; then apply the indicated channel.\\
     Take a direct sum over $S\subseteq[n]$. The weights $q_S$ are arbitrary.};
\end{tikzpicture}

%% file: figures/tetrahedral_code.tex
\begin{tikzpicture}[x=1cm,y=1cm,font=\small,>=Latex]
  \node[anchor=west,font=\bfseries\small] at (0,3.6) {Four carriers for the message $(x,y)$};
  \foreach \x/\lab/\i in {1/{|x\rangle}/1,3.3/{|y\rangle}/2,5.6/{|x\oplus y\rangle}/3,8.0/{\tau_{xy}}/4} {
    \node[draw,rounded corners,minimum width=1.8cm,minimum height=.9cm] at (\x,2.4) {$\lab$};
    \node at (\x,1.6) {carrier \i};
  }
  \begin{scope}[shift={(10.7,1.85)},scale=.83]
    \draw[gray] (0,0) circle (1.2);
    \draw[gray,dashed] (0,0) ellipse (1.2 and .35);
    \coordinate (a) at (0,1.2);\coordinate (b) at (-1,-.68);
    \coordinate (c) at (1,-.68);\coordinate (d) at (0,-.1);
    \draw[thick] (a)--(b)--(c)--cycle;
    \draw[dashed] (a)--(d)--(b);\draw[dashed] (d)--(c);
    \foreach \v in {a,b,c,d} \fill (\v) circle (2pt);
  \end{scope}
  \node[align=center] at (10.7,.25) {tetrahedral\\Bloch vectors};
  \node[draw,rounded corners,align=center,inner sep=7pt,text width=4.4cm,minimum height=1.6cm] at (2.6,-1.55)
    {Two of carriers $1,2,3$ survive:\\the message is known exactly.};
  \node[draw,rounded corners,align=center,inner sep=7pt,text width=6cm,minimum height=1.6cm] at (8.8,-1.55)
    {Carrier $4$ and one classical bit survive:\\two pure states with squared overlap $1/3$;\\
     success $\frac12(1+\sqrt{2/3})$.};
\end{tikzpicture}

%% file: figures/performance_gap.tex
\begin{tikzpicture}
\begin{axis}[
  width=.93\linewidth,height=7.0cm,
  xmin=0,xmax=1,ymin=-.002,ymax=.037,
  xlabel={Erasure probability $p$},
  ylabel={$s_{4,2,p}(4)-P_{\mathrm{succ}}$},
  xtick={0,.2,.4,.6,.8,1},ytick={0,.01,.02,.03},
  scaled ticks=false, tick label style={font=\small},
  label style={font=\small},
  legend style={font=\small,at={(.5,1.02)},anchor=south,draw=none},
  legend cell align=left,grid=major,grid style={gray!15},
  domain=0:1,samples=151]
  \addplot[black,thick,dashed] {.5*x^2*(1-x)^2};
  \addlegendentry{Optimal classical code, without feedback}
  \addplot[black,very thick] {((3-sqrt(6))/2)*x^2*(1-x)^2};
  \addlegendentry{Tetrahedral quantum code, without feedback}
  \addplot[black,thick,dashdotted] {0};
  \addlegendentry{Optimal feedback code}
\end{axis}
\end{tikzpicture}

%% file: figures/normal_rate.tex
\begin{tikzpicture}
\begin{semilogxaxis}[
width=.91\linewidth,height=7.4cm,
 tick label style={font=\small},label style={font=\small},
 legend style={font=\small,at={(.5,1.03)},anchor=south,draw=none},
 legend cell align=left,grid=major,grid style={dotted},
 scaled ticks=false,
 xmin=10,xmax=1000,ymin=.27,ymax=.515,
 xlabel={Blocklength $n$ (logarithmic scale)},ylabel={Rate (bits per channel use)},
 xtick={10,20,50,100,200,500,1000},
 xticklabels={$10$,$20$,$50$,$100$,$200$,$500$,$1000$},
 ytick={.3,.35,.4,.45,.5},
]
 \addplot[thick] table[x=n,y=rate] {figures/normal_rate.dat};
 \addlegendentry{Exact feedback-assisted rate}
 \addplot[thick,dashed] table[x=n,y=normal_rate] {figures/normal_rate.dat};
 \addlegendentry{Two-term normal approximation}
 \addplot[thick,densely dotted,domain=10:1000,samples=2] {.5};
 \addlegendentry{Limiting capacity $1/2$}
\end{semilogxaxis}
\end{tikzpicture}

%% file: figures/strong_converse.tex
\begin{tikzpicture}
\begin{axis}[
width=.91\linewidth,height=7.4cm,
 tick label style={font=\small},label style={font=\small},
 legend style={font=\small,at={(.5,1.03)},anchor=south,draw=none,/tikz/every even column/.append style={column sep=9pt}},
 legend cell align=left,grid=major,grid style={dotted},
 scaled ticks=false,
 xmin=.5,xmax=1.2,ymin=0,ymax=.66,
 xlabel={Rate $R$ (bits per channel use)},ylabel={Success exponent (bits per channel use)},
 xtick={.5,.666666667,.8,1,1.2},
 xticklabels={$1/2$,$2/3$,$0.8$,$1.0$,$1.2$},
 ytick={0,.1,.2,.3,.4,.5,.6},legend columns=2,
]
 \addplot[very thick] table[x=R,y=limit] {figures/strong_converse.dat};
 \addlegendentry{Limiting exponent}
 \addplot[thick,dashed] table[x=R,y=n20] {figures/strong_converse.dat};
 \addlegendentry{$n=20$, with feedback}
 \addplot[thick,densely dotted] table[x=R,y=n100] {figures/strong_converse.dat};
 \addlegendentry{$n=100$, with feedback}
 \addplot[thick,dashdotted] table[x=R,y=n500] {figures/strong_converse.dat};
 \addlegendentry{$n=500$, with feedback}
 \draw[densely dotted] (axis cs:2/3,0)--(axis cs:2/3,.64);
\end{axis}
\end{tikzpicture}

%% file: figures/feedback_error_exponent.tex
\begin{tikzpicture}
\begin{axis}[
width=.91\linewidth,height=7.4cm,
 tick label style={font=\small},label style={font=\small},
 legend style={font=\small,at={(.5,1.03)},anchor=south,draw=none,/tikz/every even column/.append style={column sep=9pt}},
 legend cell align=left,grid=major,grid style={dotted},
 scaled ticks=false,
 xmin=0,xmax=.5,ymin=0,ymax=1.06,
 xlabel={Rate $R$ (bits per channel use)},ylabel={Error exponent (bits per channel use)},
 xtick={0,.1,.2,.3,.4,.5},ytick={0,.2,.4,.6,.8,1},
 legend columns=2,
]
 \addplot[very thick] table[x=R,y=limit] {figures/feedback_error_exponent.dat};
 \addlegendentry{Limiting exponent}
 \addplot[thick,dashed] table[x=R,y=n20] {figures/feedback_error_exponent.dat};
 \addlegendentry{$n=20$, with feedback}
 \addplot[thick,densely dotted] table[x=R,y=n100] {figures/feedback_error_exponent.dat};
 \addlegendentry{$n=100$, with feedback}
 \addplot[thick,dashdotted] table[x=R,y=n500] {figures/feedback_error_exponent.dat};
 \addlegendentry{$n=500$, with feedback}
\end{axis}
\end{tikzpicture}

%% file: figures/reliability_bounds.tex
\begin{tikzpicture}
\begin{axis}[
 width=.88\textwidth,height=7.5cm,
 xmin=0,xmax=.5,ymin=0,ymax=1.03,
 xlabel={Rate $R$ (bits per channel use)},
 ylabel={Error exponent (bits per channel use)},
 xtick={0,.1,.2,.333333333,.5},
 xticklabels={$0$,$0.1$,$0.2$,$R_{\mathrm{crit}}=1/3$,$C=1/2$},
 tick label style={font=\small},label style={font=\small},
 legend style={font=\small,at={(.98,.98)},anchor=north east,draw=none,fill=white},
 legend cell align=left,
 grid=major,grid style={dotted},
 ]
 \addplot[thick] table[x=R,y=feedback] {figures/reliability_bounds.dat};
 \addlegendentry{Feedback exact / no-feedback upper bound}
 \addplot[thick,dashed] table[x=R,y=random] {figures/reliability_bounds.dat};
 \addlegendentry{Classical random-coding lower bound}
 \addplot[thick,densely dotted] table[x=R,y=expurgated] {figures/reliability_bounds.dat};
 \addlegendentry{Classical expurgated lower bound (clipped at zero)}
 \draw[dashdotted] (axis cs:.333333333,0)--(axis cs:.333333333,.27);
 \node[font=\small,anchor=west,align=left] at (axis cs:.344,.24) {Exact without\\feedback here};
 \draw[->] (axis cs:.384,.13)--(axis cs:.412,.025);
\end{axis}
\end{tikzpicture}